\documentclass{svproc}
\usepackage{graphicx}%
\usepackage{multirow}%
\usepackage{amsmath,amssymb,amsfonts}%
\usepackage{mathrsfs}%
\usepackage[title]{appendix}%
\usepackage{xcolor}%
\usepackage{textcomp}%
\usepackage{manyfoot}%
\usepackage{booktabs}%
\usepackage{algorithm}%
\usepackage{algorithmicx}%
\usepackage{algpseudocode}%
\usepackage{listings}%
\usepackage{changes}
\usepackage{tabularray}
\usepackage{romannum}
\usepackage[T1]{fontenc}
\usepackage{url}
\usepackage{xspace}
\usepackage[makeroom]{cancel}

\usepackage{placeins}

\usepackage{tikzscale}
\newcommand{\staticSize}{6.0cm}
\newcommand{\looseEndSize}{4.5cm}

\usepackage{todonotes}
\usetikzlibrary{positioning, 
                quotes}

\usepackage{etoolbox}
\newtoggle{report}
 \toggletrue{report}

\newcommand{\TODO}[1]{\todo[inline]{#1}}

\newcommand{\staticb}{\textsc{b-suitor}}
\newcommand{\dynb}{\textsc{Dyn-b-suitor}}
\newcommand{\laff}{{locally affected}}
\newcommand{\sat}{{saturated}}
\newcommand{\usat}{{unsaturated}}
\newcommand{\none}{\texttt{None}}
\newcommand{\edgeins}{\textsc{EdgeInsertion}$(\cdot,\cdot)$}
\newcommand{\edgeinsParam}[2]{\textsc{EdgeInsertion}(#1,#2)}
\newcommand{\edgerem}{\textsc{EdgeRemoval}$(\cdot,\cdot)$}
\newcommand{\edgeremParam}[2]{\textsc{EdgeRemoval}(#1,#2)}
\newcommand{\bedgeins}{\textsc{BatchEdgeInsertion}$(\cdot,\cdot)$}
\newcommand{\bedgerem}{\textsc{BatchEdgeRemoval}$(\cdot,\cdot)$}
\newcommand{\findaff}{\textsc{FindAffected}$(\cdot)$}
\newcommand{\uppath}{update path}

\newcommand{\drop}[1]{}

\DeclareMathOperator*{\argmaxA}{argmax}

\begin{document}

\newcommand{\ie}{i.\,e.,\xspace}
\newcommand{\st}{s.\,t.,\xspace}
\newcommand{\eg}{e.\,g.,\xspace}
\newcommand{\wrt}{w.\,r.\,t.\xspace}
\newcommand{\etc}{etc.\xspace}
\newcommand{\aka}{a.\,k.\,a.\xspace}
\newcommand{\vs}{vs.\ }
\newcommand{\etal}{et al.\ }
\newcommand{\quot}[1]{``#1''}
\newcommand{\registered}{\textsuperscript\textregistered}

\newcommand{\roundb}[1]{\left(#1\right)}
\newcommand{\sqb}[1]{\left[#1\right]}

\newcommand{\E}{{\ensuremath{\mathbb{E}}}}
\newcommand{\Oh}{\ensuremath{\mathcal{O}}}
\newcommand{\argmin}{\operatorname{argmin}}
\newcommand{\argmax}{\operatorname{argmax}}
\newcommand{\mathbc}{\mathbf{b}}
\newcommand{\exchange}[3]{#1^{+#2}_{-#3}}
\newcommand{\add}[2]{#1^{+#2}}
\newcommand{\remove}[2]{#1_{-#2}}
\newcommand{\gain}[2]{\mathrm{gain}(#1,#2)}
\newcommand{\dist}[1]{\mathrm{dist}(#1)}
\newcommand{\bfsdist}[1]{\mathrm{dist}_{\text{BFS}}(#1)}
\newcommand{\onedigit}[1]{\num[round-mode=places,round-precision=1]{#1}}
\newcommand{\integer}[1]{\num[round-mode=places,round-precision=0]{#1}}
\newcommand{\ceil}[1]{\left\lceil #1 \right\rceil}
\newcommand{\floor}[1]{\left\lfloor #1 \right\rfloor}
\newcommand{\lonerel}{\ensuremath{\mathrm{L1}_\mathrm{rel}}}
\newcommand{\ltworel}{\ensuremath{\mathrm{L2}_\mathrm{rel}}}
\newcommand{\erel}{\ensuremath{\mathrm{E}_\mathrm{rel}}}
\newcommand{\onen}[1]{\lVert #1 \rVert_1}
\newcommand{\twon}[1]{\lVert #1 \rVert_2}

\newcommand{\norm}[1]{\left\lVert#1\right\rVert}
\newcommand{\diag}[1]{\operatorname{diag}(#1)}
\newcommand{\diam}[1]{\operatorname{diam}(#1)}
\newcommand{\rad}[1]{\operatorname{rad}(#1)}
\newcommand{\mat}[1]{\mathbf{#1}}
\newcommand{\myvec}[1]{\mathbf{#1}}
\newcommand{\ment}[3]{\mathbf{#1}[#2,#3]}
\newcommand{\vent}[2]{\myvec{#1}[#2]}
\newcommand{\Lap}{\mat{L}}
\newcommand{\Lpinv}{\mat{L}^\dagger}
\newcommand{\LpinvH}{\mat{L}^{\dagger/2}}
\newcommand{\LpinvG}[1]{\mat{L}_{#1}^\dagger}
\newcommand{\nd}[2]{\myvec{D}(#1,#2)}
\newcommand{\effres}[2]{\myvec{r}(#1,#2)}
\newcommand{\effresG}[3]{\myvec{r_{#1}}(#2,#3)}
\newcommand{\trace}[1]{\operatorname{tr}(#1)}
\newcommand{\uvec}[1]{\mathbf{e}_{#1}}
\newcommand{\onesvec}{\mathbf{1}}
\newcommand{\ecc}{\operatorname{ecc}}
\newcommand{\farnel}[2]{f^{el}(#2)}
\newcommand{\farnc}[2]{f^{c}(#2)}
\newcommand{\Lab}[1]{\mathcal{L}[#1]}
\newcommand{\dir}[1]{\operatorname{direction}[#1]}
\newcommand{\idx}[2]{\textsc{Index}(#1, #2)}
\newcommand{\poly}[1]{\operatorname{poly}(#1)}
\newcommand{\numtrees}[4]{N_{#1,#4}(#2,#3)}

\newcommand{\swapMaxV}{LS-MD\xspace}
\newcommand{\swapMaxInc}{LS-MI\xspace}

\newcommand{\bekasH}{\tool{Bekas-h}\xspace}
\newcommand{\bekas}{\tool{Bekas}\xspace}
\newcommand{\jltLamg}{\tool{Lamg-jlt}\xspace}
\newcommand{\jltKyng}{\tool{Julia-jlt}\xspace}
\newcommand{\ust}{\tool{UST}\xspace}

\newcommand{\grip}{\tool{k-GRIP}\xspace}

\newcommand{\tool}[1]{\textsf{#1}}
\newcommand{\nwk}{\tool{NetworKit}\xspace}
\newcommand{\pinv}{\tool{pinv}\xspace}

\newcommand{\rec}[1]{\textcolor{gray}{[recons.: #1]}\xspace} 
\newcommand{\checkbeforesubmit}[1]{\textcolor{red}{#1 (check)}}
\newcommand{\final}[1]{\textcolor{blue}{#1\space}} 

\newcommand{\instTabColSep}{3pt}

\mainmatter              
\title{\Large A Fully-dynamic Approximation Algorithm \\ for Maximum Weight $b$-Matchings in Graphs}
\titlerunning{A Fully-dynamic Appr. Algorithm for Maximum Weight $b$-Matchings}  
%

\author{Fabian Brandt-Tumescheit 
\and Frieda Gerharz 
\and Henning Meyerhenke 
}

\authorrunning{F.\ Brandt-Tumescheit, F.\ Gerharz, and H.\ Meyerhenke} 

\institute{Humboldt-Universität zu Berlin, Germany\\
\email{\{brandtfa,meyerhenke\}@hu-berlin.de}\\ 
}

\maketitle              

\begin{abstract}
    Matching nodes in a graph $G=(V,E)$ is a well-studied algorithmic problem with many applications.
    The $b$-matching problem is a generalization that allows to match a node with up to $b$ neighbors.
    This allows more flexible connectivity patterns whenever vertices may have multiple associations.
    The algorithm \staticb\ [Khan et al., SISC 2016] is able to compute a (1/2)-approximation of a maximum weight $b$-matching in $\mathcal{O}(|E|)$ time.
    Since real-world graphs often change over time,
    fast dynamic methods for $b$-matching optimization are desirable.
    

    In this work, we propose \dynb, a dynamic algorithm for the weighted $b$-matching problem.
    As a non-trivial extension to the dynamic Suitor algorithm for $1$-matchings [Angriman et al., JEA 2022],
    our approach computes (1/2)-approximate $b$-matchings by identifying and updating
    affected vertices without static recomputation. 
    Our proposed algorithm is fully-dynamic, \ie it supports both edge insertions and deletions,
    and we prove that it computes the same solution as its static counterpart.

    In extensive experiments on real-world benchmark graphs and generated instances, our dynamic algorithm yields significant savings compared to the sequential \staticb, \eg for batch updates with $10^3$ edges with an average acceleration factor of $10^3\times$.
    When comparing our sequential dynamic algorithm with the parallel (static) \staticb\ on a 128-core machine, our dynamic algorithm is still $59\times$ to $10^4\times$ faster. 
\keywords{approximate $b$-matching, dynamic graph algorithm, algorithmic network analysis}
\end{abstract}

\section{Introduction}
\label{sec:introduction}
%
A $b$-matching $M$ in a graph $G=(V,E)$ is a subgraph of $G$ such that each vertex is incident to
at most $b$ edges in $M$. This way, the concept of $b$-matchings generalizes the
widely known notion of matchings, which are $1$-matchings ($b=1$).
%
$b$-matchings ($b\geq 1$) find applications in various fields, including social network anlaysis~\cite{HOEFER201320},
network design~\cite{lebedev2007using,bai2019coded}, scheduling~\cite{perumal2021electric}, 
resource allocation~\cite{bienkowski2021online,dong2020b,hadji2016new}, 
and fraud detection in finance networks~\cite{liu2023tracking}. 
Compared to ordinary matchings, $b$-matchings ($b \geq 2$) allow more flexible connectivity
patterns in scenarios where vertices may have multiple (yet still a restricted number of) associations.
\drop{
$b$-matchings are also used as a heuristic to solve other tasks, such as
the $\mathcal{NP}$-hard $k$-edge coloring problem~\cite{el2023b}.

\TODO{Ist das wirklich so, dass es dort als Heuristik genutzt wird?}
}



Graphs derived from real-world applications are often dynamic in nature; consider social, information, or infrastructure networks, for example.
While adding friends, comments, or likes in social networks is just one click of work, the update of a matching can become comparably time-consuming in a large graph when it is recomputed entirely from scratch.
%
To avoid a bottleneck when updating $M$ in large graphs,
one needs to restrict the work to those graph areas that are affected
by the changes in $G$, while ensuring that the new $b$-matching reflects these changes as well.
When supporting both edge insertions and deletions in this manner, we refer to an algorithm as fully-dynamic.
For the $1$-matching problem, fully-dynamic algorithms have been proposed and implemented previously for maximum weight matching (\ie a set of independent edges with maximum total
weight)~\cite{angriman2022batch,angriman2021fully,azarmehr2024fully,baswana2018fully,behnezhad2020fully,henzinger2020dynamic}.
The algorithm proposed by Angriman \etal \cite{angriman2022batch} is based on the static Suitor algorithm~\cite{manne2014new}.
Since \staticb\ by Khan \etal \cite{khan2016efficient} shares the same foundation,
it is also a natural starting point for the design of a fully-dynamic $b$-Suitor algorithm for $b$-matchings.
While there exists other work on dynamic $b$-matching algorithms, the approaches are either theoretical in nature~\cite{azarmehr2024fully,bhattacharya2017improved,el2023b} and have not been thoroughly studied in practice or they support only edge insertions~\cite{bienkowski2021online}.

\iftoggle{report}
{}
{\vspace*{-0.6\baselineskip}}

\paragraph*{Contributions.}
In this paper, we propose \dynb\ -- a fully-dynamic $b$-matching algorithm that maintains a (1/2)-approximation of a maximum weight $b$-matching (see Section~\ref{sec:dynamic_b_suitor}).
The algorithm is influenced mainly by two previous works: \staticb\ by Khan \etal \cite{khan2016efficient} (Section~\ref{sub:static_b_suitor}) and the dynamic Suitor algorithm for the $1$-matching problem~\cite{angriman2022batch} (Section~\ref{sub:dynamic_suitor}).
Starting initially with a $b$-matching derived from \staticb, our fully-dynamic algorithm finds and updates all vertices whose matching decision is potentially affected by the recent changes to the graph.
Our main theoretical result (Theorem~\ref{theorem:dynb_final}), which requires a series of lemmas, is to prove that \dynb\ computes the same result as \staticb.

According to our extensive experiments (Section~\ref{sec:experimental_results}), our \dynb\ implementation yields an algorithmic speedup (= static running time divided by dynamic one) 
of $10^6$ over static recomputation for single updates (edge insertion/deletion) and $10^3$ for batches with $10^3$ edges.
Moreover, for all 20 networks from four different categories (social, infrastructure, sparse, generated), the ratio of algorithmic speedup to batch size is slightly above $1$.
This advantage of batch processing is due to synergies when visiting and updating affected nodes and edges from the same batch.
Running times are in the millisecond range or below, making the algorithm $59\times$ to $10^4\times$ faster than the state-of-the-art parallel \staticb\ implementation~\cite{khan2016efficient} on $128$ CPU cores, depending on the batch size. 


\iftoggle{report}
{}
{\textbf{The appendix of this paper is omitted due to space constraints; it can be found in the extended version of this 
paper~\cite{extended2024bmatchingmaterial}.}}

\section{Preliminaries}
\label{sec:preliminaries}
\vspace{-0.1cm}
\subsection{Problem Definition and Notation}
\label{subsec:definition}
\vspace{-0.1cm}
%
We consider undirected weighted graphs $G = (V,E,w)$ with $n := |V|$ vertices (= nodes), $m := |E|$ edges,
and the edge weight function $\omega: E \rightarrow \mathbb{R}_{>0}$.
The set of neighbors of a vertex $u$ is denoted by $N(u)$ and the (unweighted) degree of $u$ as $\text{deg}(u) := |N(u)|$.
We use a common tie-breaking rule to define a total ordering: in case the weight of two edges $\{u, t\}$ and $\{u, v\}$ is equal, we say that $\omega(u, t) > \omega(u, v)$ if $t < v$, \ie if the node ID of $t$ is smaller than the ID of $v$.
This is required for breaking ties when selecting the heaviest
neighbor of $u$. 

A $b$-matching is a set of edges $M \subseteq  E$ in which no
more than $b(u)$ edges are incident to $u$, where $b(\cdot)$ is a function that maps
every vertex $u$ to a natural number.
Thus, $b$-matchings are a generalization of ordinary matchings; the latter are a special case by 
using the constant function $b : u \rightarrow 1$.
%
%
The weight of a $b$-matching $M$ is defined as $\omega(M) = \sum_{e \in M} \omega(e)$.
If $M$ has maximum weight among all $b$-matchings, it is called a maximum weight $b$-matching (MWBM).
This paper deals with finding an approximate MWBM.
An edge $\{u,v\}$ is called \textit{matched} if it belongs to $M$; otherwise it is called \textit{unmatched}.
If a vertex $u$ is incident to $b$ edges from $M$, then
the vertex is called \textit{saturated}. If not, it is \textit{unsaturated}.
In general, different vertices can have different saturation goals $b(\cdot)$.
In our setting, these goals do not change when the graph changes over time.

\subsection{Related Work}
\label{subsec:related_work}

The MWBM problem has been studied both in terms of static and dynamic settings. We give an overview over both; afterwards, we explain two algorithms with high relevance for this paper in more detail, \staticb\ for $b$-matchings~\cite{khan2016efficient} and the dynamic Suitor algorithm for $1$-matchings~\cite{angriman2022batch}.
For the dynamic $1$-matching problem, a wide spectrum of algorithms have been proposed in the literature;
we refer the interested reader to Ref.~\cite{angriman2021fully} for an overview. 

Exact algorithms for MWBM have been around since the 1960s.
%
\drop{
Edmonds~\cite{edmonds1965maximum} proposed one with time complexity $\Oh(n^3)$. 
Further improvements over the years could not decrease the time complexity below $\Oh(|V|\cdot |E|)$~\cite{pulleyblank1973,gabow1983efficient,muller2000implementing,gabow2018data}. 
The same time complexity even applies if the problem is restricted to bipartite graphs~\cite{fremuth1999balanced,huang2011fast,huang2007loopy}.
A natural construction for solving the $b$-matching problem is to apply reductions in order to arrive at $1$-matching~\cite{shiloach1981degree,huang2017approximate,ferdous2021degree}.
However, this introduces at least $b \times |V|$ additional vertices, substantially increasing the graph size.
}
Overall, the time complexity for computing the exact MWBM solution using known techniques
is $\Oh(|V||E|)$~\cite{gabow2018data}, which can become impractical for very large instances.
As a result, several approximation algorithms have been proposed.

Huang and Pettie~\cite{huang2017approximate} proposed a $(1-\epsilon)$ approximation algorithm for the MWBM problem 
by approximating the conditions for a relaxed complementary slackness formulation of the original problem.
While the running time is nearly linear in $|E|$ using a scaling framework~\cite{duan2014linear} for sufficiently large values of $\epsilon$, the algorithm itself is theoretical in nature.
Subsequently, Samineni \etal\cite{samineni2024approximate} addressed the MWBM problem by formulating an LP-relaxation based on dual variables. 
Compared to Huang and Pettie, their approach is more feasible for an actual implementation but introduces an additional dependency on $b$ in the time complexity.
However, the algorithm is only suitable for bipartite graphs, where vertices are split in \textit{bidders} and \textit{objects} in order to define an auction-like setup.

Another set of approximation algorithms is based on local optimization.
This includes a $(2/3-\epsilon)$ randomized algorithm~\cite{mestre2006greedy} that picks vertices uniformly at random and tries to match a neighbor based on the maximum gain between matched and unmatched neighbors.
Its running time depends on $\log{(1/\epsilon)}$, $b$ and the number of edges; moreover, due to its randomized nature,
its result is not deterministic.
Several $(1/2)$-approximation algorithms are based on the \textsc{Greedy} algorithm by Mestre~\cite{mestre2006greedy}.
They remove the running time dependency on $\epsilon$ and produce identical $b$-matchings~\cite{khan2016efficient,drake2003simple,pothen2019approximation}.
The $b$-suitor algorithm by Khan \etal\cite{khan2016efficient} proves to be the fastest technique within the family of greedy algorithms. 
According to empirical results~\cite{khan2016efficient}, the solution quality can be much better than the proven $(1/2)$-approximation, with a margin of only $\sim1\%$ from the optimal value for some real-world networks.
The algorithmic construction makes the $b$-suitor algorithm also suitable for parallelization. 
While the original publication already includes a parallel variant for shared memory~\cite{khan2016efficient}, more recent publications adapt the construction to GPUs~\cite{naim2018scalable} and distributed environments~\cite{khan2016designing}.
Since our dynamic algorithm uses the $b$-suitor algorithm initially and computes the same matchings, we explain the latter in more detail in Section~\ref{sub:static_b_suitor}.

Concerning dynamic algorithms for MWBM on general graphs, the contributions so far have been theoretical in nature~\cite{bhattacharya2017improved,azarmehr2024fully,el2023b}.
Also for bipartite graphs, existing algorithms have yet to be implemented~\cite{bernstein2015fully,ting2015near}.
One exception is a reinforcement learning approach by Wang \etal \cite{wang2019adaptive}, which is concerned with inserting single nodes and corresponding edges into the (bipartite) graph together at the same time. 
Its idea is to use small batches of inserted node-edge tuples together and do a greedy optmization locally inside the batch.
The optimal batch splitting is then learned from the Hungarian algorithm.
So far, however, there is no way known to transfer this approach to MWBM on general graphs.
To the best of our knowledge, our proposed algorithm is the first implemented fully-dynamic algorithm
for the (approximate) MWBM problem for general graphs.

\drop{One algorithm to mention here is by Angriman \etal \cite{angriman2022batch}; it is based on the static suitor algorithm~\cite{manne2014new} and supports edge insertions and deletions.
Since the $b$-suitor algorithm is an adaptation of the suitor algorithm for $b$-matchings,
the dynamic framework~\cite{angriman2022batch} is also a good starting point for designing a dynamic 
algorithm for $b$-matchings. Hence, we will explain the framework in Section~\ref{sub:dynamic_suitor}.
}

\subsection{The \staticb\ Algorithm}
\label{sub:static_b_suitor}
\vspace{-0.1cm}
\staticb\ by Khan \etal \cite{khan2016efficient} is a (1/2)-approximation algorithm; it extends the 
Suitor algorithm~\cite{manne2014new} for matchings to $b \geq 2$.
In its core, \staticb\ iterates over all nodes $u$ of $G$ and tries to find up to $b(u)$ matching partners for $u$. Conceptually, each node negotiates with its neighbors to form a shared proposal, referred to as suitors.
To this end, for each node there exists a min-priority queue $S(u)$ with a maximum size of $b(u)$, representing all neighbors of $u$ which have already proposed to $u$ as potential suitors.
If $S(u)$ contains $b(u)$ elements, we call it \sat\ and $S(u).min$ returns the suitor with the lowest edge weight.
In case $|S(u)|<b(u)$, $S(u)$ is \usat\ and the result of $S(u).min$ is $\texttt{None}$.

While the original pseudocode includes an additional helper set $T(u)$ that contains intermediate matches until the algorithm finishes, adaptations~\cite{naim2018scalable,khan2016designing} remove this set.
In this adapted version, the invariant $v \in S(u) \Leftrightarrow u \in S(v)$ is maintained throughout the execution of the algorithm.
In the following, we refer to this invariant as the \textit{S-invariant} for convenience.

When the main vertex of the current iteration is $u$, the goal is to find a matching partner $p_u$ for $u$
according to:
\begin{equation}
    p_u=\argmaxA_{v \in N(u) \setminus S(u)} \{ \omega\left(u,v\right) \vert ~\omega\left(u,v\right) > \omega(v, S(v).min) \}
    \label{eq:static_suitor}
\end{equation}
If a matching partner of this kind is found, $u$ updates its list of suitors $S(u)$ and also adds itself to the suitor list of that partner.
If that partner was saturated before, we need to remove its previous minimum suitor and find a new partner for the removed vertex recursively (also see the iterative version as Algorithm~\ref{alg:static_suitor_main} in Appendix~\ref{sub:appendix-staticb}).
The key result about \staticb\ is:
\begin{proposition}[\cite{khan2016efficient}]
    The $S$-queues of all nodes form a valid $b$-matching $M=\{\{u,v\} \in E~ \vert~ u \in S(v), v \in S(u)\}$.
    $M$ is a (1/2)-approximation of a (not necessarily unique) maximum weighted $b$-matching.
    \label{proposition:valid_matching_static}       
\end{proposition}

The overall running time of \staticb\ is $\Oh(m\Delta\log{\beta})$, with $\Delta$ as maximum vertex degree of $G$ and $\beta=\max_{v \in V}{b(v)}$. 
Obviously, by setting $b=1$, \staticb\ can also be used for solving the corresponding $1$-matching problem.
Khan et al.~\cite[Sec.~3]{khan2016efficient} also design a parallel version of the algorithm for shared-memory systems.
It uses a global queue and vertex-specific locks whenever $S(u)$ is updated.
Since our dynamic algorithm is sequential in nature, we mainly compare against the sequential \staticb.
Note that our experiments also show that \dynb\ is much faster than recomputing
from scratch in parallel. 

\subsection{The Dynamic Suitor Algorithm}
\label{sub:dynamic_suitor}
\vspace{-0.1cm}
For the 1-matching problem, Angriman et al.~\cite{angriman2022batch} design a fully-dynamic algorithm based on the Suitor algorithm by Manne et al.~\cite{manne2014new}.
Starting with a matching $M$ in the initial graph, \eg computed by the static Suitor algorithm, the main idea of the dynamic matching (DM) algorithm is to update the matching by identifying all vertices influenced by the changes in the graph.
For this to work, Angriman \etal introduce the definition of \textit{affected} vertices~\cite[Def.~1]{angriman2022batch}.
The main idea is to call any vertex $u$ {affected} for which there exists a better suitor in the neighborhood $N(u)$. Since $b=1$, $S(u)$ now refers to a single suitor instead of a list. 
As a result, affected vertices can be found iteratively by repeatedly applying Eq.~(\ref{eq:dynamic_suitor}) -- a variation of Eq.~(\ref{eq:static_suitor}):
\begin{equation}
    s_u=\argmaxA_{v \in N(u) \setminus S(u)} \{ \omega\left(u,v\right) \vert ~\omega\left(u,v\right) > \omega(v, S(v)) \}
    \label{eq:dynamic_suitor}
\end{equation}

Whenever $s_u$ is not \texttt{None}, $u$ is affected and we connect $u$ and $s_u$ as new suitors. The iterative process continues by checking Eq.~(\ref{eq:dynamic_suitor}) again for the previous suitor of $s_u$. This process continues until no new affected vertex is found (\ie Eq.~(\ref{eq:dynamic_suitor}) returns \texttt{None} for the vertex in question).

The iterative pursuit of locally affected vertices traverses update paths for both edge insertions and deletions.
Each path starts at one endpoint $u$ or $v$, respectively, of the updated edge $\{u,v\}$.
Angriman \etal prove that after handling all vertices on the path, the resulting matching equals the matching computed by the static Suitor algorithm.
In the worst-case, this iterative process takes $\Oh(n+m)$ time. Especially for single edge updates, only a fraction of nodes is affected, though. 
Even for batch updates with up to $10^3$ nodes, the algorithmic speedup (DM vs Suitor) is still several orders of magnitude.

Note that the rules for identifying affected vertices in their \textsc{FindAffected} routine~\cite[Alg. 1]{angriman2022batch} cannot be applied or directly translated to the $b$-matching case, since nodes might need to be updated multiple times if $b>1$.
Moreover, the dynamic Suitor algorithm traverses the update paths twice\drop{(forward in function \textsc{FindAffected} and backwards in function \textsc{UpdateAffected})}, leaving room for improvement.

\section{The \dynb\ Algorithm}
\label{sec:dynamic_b_suitor}
\begin{table*}[b]
    \begin{center}
        \begin{tabular}{||l c c c c||}
            \hline
            \textbf{state}                                                                             & \textbf{graph}                  & \textbf{neighbors} & \textbf{suitors}  & \textbf{$b$-matching} \\ [0.5ex]
            \hline\hline
            \begin{tabular}{@{}l@{}}original graph + \\ after \staticb \end{tabular}           & $G=(V,E)$                       & $N(u)$             & $S(u)$            & $M$                   \\
            \hline
            \begin{tabular}{@{}l@{}}updated graph + \\ after \staticb \end{tabular}            & $\widetilde{G}=(V,\widetilde{E})$ & $\widetilde{N}(u)$  & $\widetilde{S}(u)$ & $\widetilde{M}$        \\
            \hline
            \begin{tabular}{@{}l@{}}updated graph + \\ \underline{during} \dynb \end{tabular} & $\widetilde{G}=(V,\widetilde{E})$   & $\widetilde{N}(u)$  & $S^{(i)}(u)$      & $M^{(i)}$             \\
            \hline
            \begin{tabular}{@{}l@{}}updated graph + \\ \underline{after} \dynb \end{tabular}  & $\widetilde{G}=(V,\widetilde{E})$   & $\widetilde{N}(u)$  & $S^{(f)}(u)$      & $M^{(f)}$             \\ [1ex]
            \hline
        \end{tabular}
    \end{center}
    \caption{Notation used for dynamic maximum weight $b$-matching computations. 
    }
    \label{tab:dynamic_notation}
\end{table*}

The core ideas of our dynamic algorithm \dynb\ extend and generalize the sequential \staticb\ as well as the DM framework 
(Sections~\ref{sub:static_b_suitor} and~\ref{sub:dynamic_suitor}). 
The resulting algorithm computes a (1/2)-approximate maximum weight $b$-matching for dynamic graphs
by revising vertices affected by an edge update.
It can handle single edges and batches of edges to be inserted or removed.

Due to overlapping design decisions, we reuse notation from Angriman \etal \cite{angriman2022batch}.
For the different states of the graph and our data structures (i) before, (ii) during, and (iii)
 after the manipulation of the graph and the execution of the algorithm, we use the notation in
 Table~\ref{tab:dynamic_notation}. 
Since the dynamic algorithm supports both edge insertions and deletions, $\widetilde{G}=(V,\widetilde{E})$, and $\widetilde{N}(u)$ denote the state of the graph data structure after edges have been added/removed.
$S$/$\widetilde{S}$ follows the definitions from Section~\ref{sub:static_b_suitor}, representing the data structures for the suitors 
\wrt the original/updated graph.
For \dynb, we introduce subscripts similar to the dynamic Suitor algorithm.
They refer to state $^{(i)}$ (intermediate) directly after the graph edge update (addition or removal) and $^{(f)}$ (final) after the dynamic algorithm has been run.
The intermediate $b$-matching $M^{(i)}$ considers the updated graph data structure; in this state, it may differ
from the matching $\widetilde{M}$ obtained by the static algorithm.
The deviation of $M^{(i)}$ with respect to $\widetilde{M}$ could either be caused by an edge insertion or removal and the associated change of suitor lists.


The goal of the dynamic algorithm is to efficiently update $S^{(i)}(u)$ after an edge insertion or removal in order to arrive at a final state $S^{(f)}(u)$, so that the invariant $v \in S^{(f)}(u) \Leftrightarrow u \in S^{(f)}(v)$ holds and also that the resulting $M^{(f)}$ equals the $b$-matching $\widetilde{M}$ computed by \staticb.
To achieve this goal, we define the notion of a \textit{locally affected} vertex, extending Eq.~(\ref{eq:static_suitor}):

\begin{definition}
    A vertex $u \in V$ is called \laff\ iff there exists a better suitor in the updated neighborhood
    $\widetilde{N}(u)$, \ie for $s_u = \argmax_{v \in \widetilde{N}(u) \setminus S^{(i)}(u)} \{ \omega\left(u,v\right) \vert$ $ ~\omega\left(u,v\right) > \omega\left(v, S^{(i)}(v).min\right) \}$ it holds that $s(u) \neq \texttt{None} \wedge \omega\left(u, S^{(i)}(u).min\right) < \omega\left(u, s_u\right)$.
    \label{def:affected}
\end{definition}

\subsection{Finding Locally Affected Vertices}
\label{subsec:locally_affected}
After a single edge update, our dynamic algorithm computes $M^{(f)}$ by finding all {locally affected} vertices iteratively, starting with
endpoints of the updated edge (see \textsc{FindAffected} as Algorithm~\ref{alg:dynamic_suitor_find_affected_new}).
In that sense, the main loop (lines \ref{algline:findaffected_main_loop_start} to \ref{algline:findaffected_main_loop_end}) reuses ideas from the dynamic suitor algorithm. 
The main difference in \dynb\ is the application of the {S-invariant} from \staticb\ throughout the dynamic update, \ie consequently updating $S^{(i)}$ for every update of node $u$ and its new partner (lines \ref{algline:s-invariant1}+\ref{algline:s-invariant2}).
Maintaining the {S-invariant} also enables \dynb\ to remove the need for traversing the update paths twice, thereby reducing the workload with respect to the dynamic suitor algorithm.

\begin{algorithm}[bt]
    \begin{algorithmic}[1]
        \begin{scriptsize}
            \Function{FindAffected}{$x$}
            \State $\text{done} \gets \texttt{False}$
            \State ${cu} \gets x$ \Comment{Init \textit{current} node} \label{algline:init-cu}
            \State ${ca} \gets S^{(i)}({cu}).min$ \Comment{Init \textit{candidate} node}
            \State $looseEnd \gets \texttt{None}$
            \Repeat \label{algline:findaffected_main_loop_start}
            \State $\text{done} \gets \texttt{True}$
            \State $ca \gets \underset{v \in \widetilde{N}({cu}) \setminus S^{(i)}({cu})}{\argmax} \{ \scriptscriptstyle \omega({cu},v) : \omega({cu},v) > \omega(v, S^{(i)}(v).min) \}$ \label{algline:check_def_affected}
            \If{ $\text{ca} == \texttt{None}$ } \Comment{Def. \ref{def:affected} not fulfilled $\rightarrow$ stop here} 
            \State \textbf{break} \label{algline:break}
            \EndIf
            \State $prevCu \gets S^{(i)}(cu).insert(ca)$ \Comment{Returns replaced item} \label{algline:s-invariant1}
            \State $prevCa \gets S^{(i)}(ca).insert(cu)$ \label{algline:s-invariant2}

            \If{$prevCu \neq \texttt{None}$} \Comment{Detect special case}\label{algline:special}
            \State $S^{(i)}(prevCu).remove(cu)$
            \State $looseEnd \gets \text{prevCu}$
            \EndIf \label{algline:special_end}

            \If{$prevCa \neq \texttt{None}$} \Comment{Found new \textit{affected} node} \label{algline:prevca_not_none}
            \State $S^{(i)}(prevCa).remove(ca)$ \Comment{Maintain \textit{S-invariant}} \label{algline:prevca_s_invariant}
            \State $\text{cu} \gets prevCa$ \Comment{ Set \textit{current} to new affected node} \label{algline:prevca_cu}
            \State $\text{done} \gets \texttt{False}$
            \EndIf
            \State $\text{ca} \gets S^{(i)}(cu).min$
            \State $\text{cw} \gets \omega(cu, ca)$ \label{def:end_loop}
            \Until{$\text{done}$} \Comment{Repeat loop until path is finished} \label{algline:findaffected_main_loop_end}
            \If{ $looseEnd$ != $\texttt{None}$}
            \State $\textsc{FindAffected}(looseEnd)$ \Comment{Repair special case} \label{algline:findaffected_loose_end}
            \EndIf
            \EndFunction
        \end{scriptsize}
    \end{algorithmic}
    \caption{Function \textsc{FindAffected}}
    \label{alg:dynamic_suitor_find_affected_new}
\end{algorithm}

\textsc{FindAffected} is called for each endpoint of an edge update, initialized as $cu$ in the algorithm (line~\ref{algline:init-cu}).
The goal is then to find a suitable candidate $ca$ (line \ref{algline:check_def_affected}).
To this end, the unmatched neighborhood $\widetilde{N}({cu}) \setminus S^{(i)}({cu})$ of $cu$ is searched.
If no eligible candidate is found or the candidate does not fulfill Definition~\ref{def:affected}, \textsc{FindAffected} exits the loop (line \ref{algline:break}).
Otherwise, we continue and match $cu$ and $ca$, maintaining the \textit{S-invariant}.
If $ca$ is saturated, \ie $prevCa$ is assigned to a not \texttt{None} value (lines~\ref{algline:s-invariant2} and~\ref{algline:prevca_not_none}), $ca$ will also be removed from $prevCa$.
Since $prevCa$ is not $saturated$ anymore and therefore might also fulfill Definition \ref{def:affected}, the next iteration of the main loop uses $prevCa$ as the new $cu$ vertex (line~\ref{algline:prevca_cu}).
Similarly to the DM framework from Ref.~\cite{angriman2022batch}, \textsc{FindAffected} traverses a so-called \uppath\ of affected nodes.
This path ends when either $ca$ was not saturated before the update or no suitable $ca$ can be found.

Compared to \staticb, \textsc{FindAffected} directly updates both suitor queues of the affected node $cu$ and the new partner $ca$.
To compensate for the possibility of recurring node updates, we introduce the concept of a \emph{loose end}. 
A loose end is identified whenever there is a \laff\ node that is already \sat\ (line \ref{algline:special}).
After the iterative update path has been finished, this loose end is processed in a second recursive run of \textsc{FindAffected}.

For \dynb\ we additionally define two more functions: \edgeins\ and \edgerem. 
Given an inserted edge $e=\{u,v\}$, \edgeinsParam{$u$}{$v$} checks whether $u$ and $v$ are \laff\ by the update, \ie whether the new edge should be in the final matching $\widetilde{M}$.
Analogously, \edgeremParam{$u$}{$v$} checks whether the removal of $e$ influences the suitor lists of $u$ and $v$ and therefore might induce changes to $\widetilde{M}$.
If an edge update indeed renders $u$ and $v$ \laff, \edgeinsParam{$u$}{$v$} and \edgeremParam{$u$}{$v$} update $S^{(i)}(u)$ and $S^{(i)}(v)$ and call \findaff\ for both $u$ and $v$, respectively, in order to identify and update all \laff\ vertices.

In addition to the single edge insertion and removal, the algorithm can also be generalized to support multiple edge updates $B=\{e_1, \dots, e_i\}$.
For both variants (removal, insertion) the main idea is to apply their single edge update counterpart in a loop on $B$ and
$\widetilde{G}$.
$\widetilde{G}$ is the updated graph based on $G$ with either multiple removed ($\widetilde{G}=(V,\widetilde{E}=E \setminus B)$) or added edges ($\widetilde{G}=(V,\widetilde{E}=E \cup B)$).
Note that the result of one loop iteration (\ie the insertion/removal of a single edge) is not the final matching $M^{(f)}$
anymore, but instead the intermediate result $M^{\left(i\right)}$.
Once the loop over all edge updates is completed, $M^{\left(i\right)}$ becomes $M^{(f)} = \widetilde{M}$.

The main theoretical result regarding \dynb\ is that, after applying the update routines, the resulting matching $M^{(f)}$
indeed equals $\widetilde{M}$, the $b$-matching computed by \staticb:

\begin{theorem}
    Let $\widetilde{M}$ be the matching computed using the \staticb\ on $\widetilde{G}$ and $B=\{e_1, \dots, e_i\}$ be an batch of edge updates on $G$. After \dynb\ is finished, $M^{(f)}$ equals $\widetilde{M}$.
    $\widetilde{M}$ is a (1/2)-approximation of the MWBM problem.
    \label{theorem:dynb_final}
\end{theorem}
\vspace{-0.2cm}
\begin{proof}
Due to space constraints, we provide the proof, based on a series of auxiliary results in Appendices~\ref{sub:appendix-edge-insertion} to~\ref{sub:appendix_batch_updates}, in Appendix~\ref{sub:appx-proof-main-thm}.
\end{proof}




\section{Experiments}
\label{sec:experimental_results}

The implementation of the dynamic $b$-suitor algorithm \dynb\ is done in C++, using the open-source framework NetworKit~\cite{angriman2023algorithms}. We plan to integrate the code into the next NetworKit release
after paper acceptance.
\iftoggle{report}
{}
{\footnote{Code, reproducibility instructions, and experimental results are available on Github: \url{https://github.com/hu-macsy/cn-2024-dynamic-b-suitor-experiments
}}}

Since \dynb\ is sequential per design, the main competitor for comparison is the sequential implementation of \staticb.
The main comparison is done by tracking the algorithmic speedup, \ie the running time of one complete run of \staticb\ 
divided by the running time of the update routine of \dynb\ after a certain number of edge updates.
To minimize potential running time differences caused by different toolkits, we implement \staticb\ in NetworKit, too.
\iftoggle{report}
{The code and reproducibility instructions can be found here: \url{https://github.com/hu-macsy/cn-2024-dynamic-b-suitor-experiments}.}
{}

\subsection{Settings}
\label{sub:exp_settings}
\vspace{-0.1cm}
All main experiments for the sequential comparison are executed on a machine with a dual-socket Intel Xeon 6126 processor with 12 cores and 192 GB RAM.
The parallel \staticb\ by Khan \etal\cite{khan2016efficient} is executed on an AMD Epyc 9754 with 128 cores and 192 GB of RAM.
To ensure reproducibility, the command line tool Simexpal~\cite{angriman2019guidelines} is used to automate the experiments.

The datasets include a wide variety of graphs that can be categorized into four different types: sparsely connected graphs (\emph{Sparse}), infrastructure networks (\emph{Infra}), social networks (\emph{Social}) and randomly generated (\emph{Random}) graphs.
Overall, these categories contain 20 different networks, ranging from $10^4$ to $10^7$ nodes and $10^6$ to $10^8$ edges. The networks are also chosen to have different underlying properties regarding average node degree.
See Appendix~\ref{sub:instance-stats} for more detailed information about the datasets. 

The batch size, the type of operation, and the value of $b$ form the major axes of the experimental investigation.
Our tests are conducted with batch sizes ranging from 1 to $10^3$ edges.
The saturation goal $b$ is taken from the set $\{1,3,10\}$, alongside randomly generated values
between $1$ and $10$, using a uniform distribution.
For each cross product of said axes (\eg graph = belgium, operation = insert, batch size = 100, $b$ = 3), \dynb\
 is repeated 50 times.
We use violin plots for visualizing the experimental data distributions.
Also, the correctness of the \dynb\ implementation was checked successfully
by comparing its $b$-matchings with those of \staticb\ on $10^6$ random $G(n,p)$ graphs.

\subsection{Evaluation}
\label{sub:exp_results_speedup}
\vspace{-0.1cm}

Since the results are similar for edge insertions and edge removal concerning the algorithmic speedup, 
we show in this section only the visualization of edge insertion and removal combined for batches of size $10^3$, see Figure~\ref{fig:res_speed}. 
The dynamic algorithm yields an average algorithmic speedup of $4.3\cdot 10^3$ over all categories.
The best results are achieved for infrastructure networks with an average algorithmic speedup of $6.5\cdot 10^3$ 
for single-edge updates, whereas the lowest is the \emph{Sparse} category with $2.1\cdot 10^3$. 
This is in line with how the algorithm works. The chosen infrastructure networks have the lowest average node degree.
As a result, even for low values of $b$, \staticb\ is more likely to revisit nodes.

\begin{figure}[t!]
    \centering
    \includegraphics[width=\linewidth]{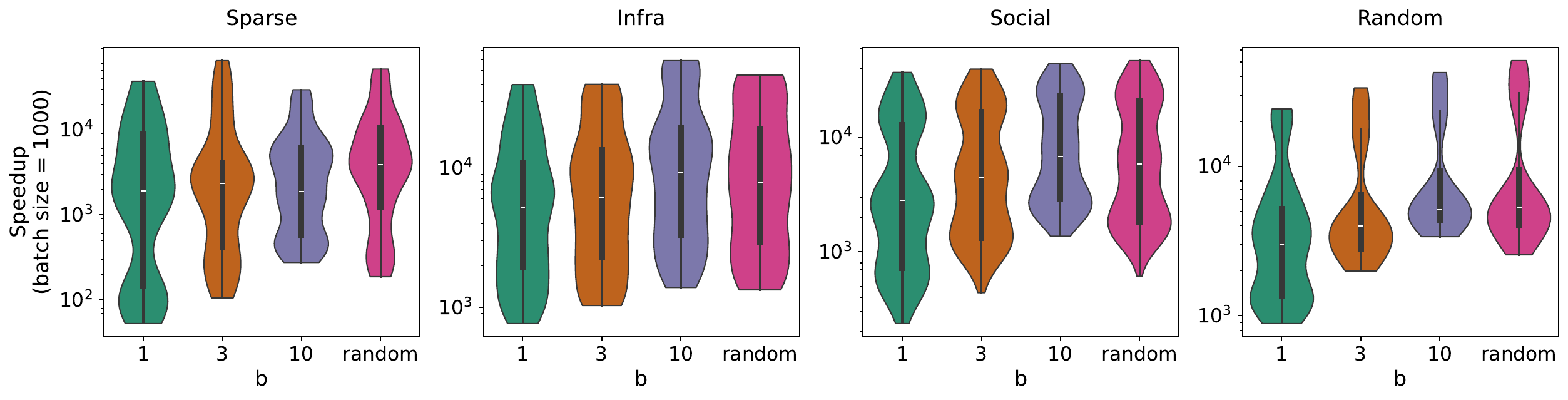}
    \caption{Speedup for batches of size $10^3$ (\textbf{edge insertion} and \textbf{edge removal} combined). For all analysed networks, the average speedup increases with the value of $b$.}
    \label{fig:res_speed}
\end{figure}

In Appendix~\ref{sub:add-exp-results} we also break down the different operations into edge insertion and 
removal and give a more detailed overview of the results with different batch sizes.
As expected, the maximum algorithmic speedup is achieved for single-edge updates:
the combined average (geometric mean) algorithmic speedup is six orders of magnitude.
When processing a batch of size $|B|$, one may expect a slowdown of $|B|$ \wrt to the running time for batch size $1$.
As can be seen from Figures~\ref{fig:res_speed_insertion_misc} and~\ref{fig:res_speed_removal_misc},
the average slowdown for batch size $10^3$ is $980$, very close to expectation.
This slight improvement can be explained by overlaps of the updates, 
\ie a certain update from one edge is already covered by the update of a previous edge in the same batch.






Furthermore, the algorithmic speedup of \dynb\ against the \emph{parallel} implementation of \staticb\ (on $128$ cores) still shows average speedup of $59\times$ for batch sizes of $10^3$, see Appendix~\ref{sub:exp_results_parallel}. 
\iftoggle{report}
{At the same place, the interested reader can find a more detailed comparison with the {parallel} \staticb.}
{}

For $b=1$, \ie an ordinary matching, we can compare our algorithmic speedups with those published in Angriman \etal\cite{angriman2022batch} directly.
Since \dynb\ iterates over \laff\ nodes only once, one might expect a measurable acceleration compared to the dynamic suitor algorithm. 
However, this is not the case; both algorithms yield the same order of magnitude. 
A reason for this deviation from expectation likely stems from line \ref{algline:check_def_affected} in \textsc{FindAffected}: it involves more time-consuming checks compared to the dynamic Suitor algorithm. 

\vspace{-0.4cm}
\section{Conclusions}
\label{sec:conclusions} 
\vspace{-0.25cm}
Our experimental results show that the \dynb\ algorithm outperforms a recomputation from scratch with the sequential \staticb\ by three to six orders of magnitudes. 
Compared to the latter's parallel variant on 128 CPU cores, \dynb\ is still at least $59\times$ faster -- depending on the batch size even up to three orders of magnitudes.
From a theoretical viewpoint, our work is a non-trivial extension of \staticb, the dynamic Suitor algorithm, and their analysis.
We could eventually prove that the resulting matching equals the one by the static algorithm and is thus a (1/2)-approximation. 
In practice, according to previous work, the weight of these $b$-matchings is often near the optimal one.
Overall, our proposed algorithm provides a significant acceleration in all cases for updating $b$-matchings with high quality
in a dynamic setting.
\vspace{-0.3cm}
\paragraph{Acknowledgements:} This work is partially supported by German Research Foundation (DFG) grant GR 5745/1-1 (DyANE).

\bibliographystyle{splncs04} 
\bibliography{dyn-b-suitor} 

\begin{thebibliography}{10}
\providecommand{\url}[1]{\texttt{#1}}
\providecommand{\urlprefix}{URL }
\providecommand{\doi}[1]{https://doi.org/#1}

\bibitem{angriman2022batch}
Angriman, E., Boro{\'n}, M., Meyerhenke, H.: A batch-dynamic suitor algorithm for approximating maximum weighted matching. ACM J. of Experimental Algorithmics (JEA)  \textbf{27}(1),  1--41 (2022)

\bibitem{angriman2023algorithms}
Angriman, E., van~der Grinten, A., Hamann, M., Meyerhenke, H., Penschuck, M.: Algorithms for large-scale network analysis and the {N}etwor{K}it toolkit. In: Algorithms for Big Data: DFG PP 1736, pp. 3--20. Springer Nature (2023)

\bibitem{angriman2019guidelines}
Angriman, E., van~der Grinten, A., von Looz, M., Meyerhenke, H., N{\"o}llenburg, M., Predari, M., Tzovas, C.: Guidelines for experimental algorithmics: A case study in network analysis. Algorithms  \textbf{12}(7), ~127 (2019)

\bibitem{angriman2021fully}
Angriman, E., Meyerhenke, H., Schulz, C., U{\c{c}}ar, B.: Fully-dynamic weighted matching approximation in practice. In: Appl. and Comp. Discrete Algorithms (ACDA). pp. 32--44. SIAM (2021)

\bibitem{azarmehr2024fully}
Azarmehr, A., Behnezhad, S., Roghani, M.: Fully dynamic matching: -approximation in polylog update time. In: Proc. {ACM-SIAM} Symp. on Discrete Algorithms {(SODA)}. pp. 3040--3061. {SIAM} (2024)

\bibitem{bai2019coded}
Bai, B., Li, W., Wang, L., Zhang, G.: Coded caching in fog-ran: $b$-matching approach. Transactions on Computers (ToC)  \textbf{67}(5),  3753--3767 (2019)

\bibitem{baswana2018fully}
Baswana, S., Gupta, M., Sen, S.: Fully dynamic maximal matching in {$\Oh(\log n)$} update time (corrected version). Journal on Computing (JoC)  \textbf{47}(3),  617--650 (2018)

\bibitem{behnezhad2020fully}
Behnezhad, S., Lkacki, J., Mirrokni, V.: Fully dynamic matching: Beating 2-approximation in {$\Delta^\epsilon$} update time. In: Symposium on Discrete Algorithms (SODA). pp. 2492--2508. SIAM (2020)

\bibitem{bernstein2015fully}
Bernstein, A., Stein, C.: Fully dynamic matching in bipartite graphs. In: Intl. Coll. on Automata, Languages and Progr. (ICALP). pp. 167--179. Springer (2015)

\bibitem{bhattacharya2017improved}
Bhattacharya, S., Gupta, M., Mohan, D.: {Improved Algorithm for Dynamic b-Matching}. In: 25th Europ. Symp. on Algorithms (ESA). LIPIcs, vol.~87, pp. 15:1--15:13. Schloss Dagstuhl -- Leibniz-Zentrum f{\"u}r Informatik (2017)

\bibitem{bienkowski2021online}
Bienkowski, M., Fuchssteiner, D., Marcinkowski, J., Schmid, S.: Online dynamic b-matching: With applications to reconfigurable datacenter networks. Special Interest Group for the computer performance evaluation community (SIGMETRICS)  \textbf{48}(3),  99--108 (2021)

\bibitem{davis2011university}
Davis, T.A., Hu, Y.: The university of florida sparse matrix collection. Transactions on Mathematical Software (TOMS)  \textbf{38}(1),  1--25 (2011)

\bibitem{dong2020b}
Dong, L., Kang, X., Pan, M., Zhao, M., Zhang, F., Yao, H.: B-matching-based optimization model for energy allocation in sea surface monitoring. Energy  \textbf{192},  116618 (2020)

\bibitem{drake2003simple}
Drake, D.E., Hougardy, S.: A simple approximation algorithm for the weighted matching problem. Information Processing Letters (IPL)  \textbf{85}(4),  211--213 (2003)

\bibitem{duan2014linear}
Duan, R., Pettie, S.: Linear-time approximation for maximum weight matching. J. ACM  \textbf{61}(1) (2014)

\bibitem{el2023b}
El-Hayek, A., Hanauer, K., Henzinger, M.: On $b$-matching and fully-dynamic maximum $ k $-edge coloring. arXiv:2310.01149  (2023)

\bibitem{gabow2018data}
Gabow, H.N.: Data structures for weighted matching and extensions to b-matching and f-factors. Transactions on Algorithms (TALG)  \textbf{14}(3),  1--80 (2018)

\bibitem{hadji2016new}
Hadji, M., Djenane, N., Aoudjit, R., Bouzefrane, S.: A new scalable and energy efficient algorithm for vms reassignment in cloud data centers. In: Future Internet of Things and Cloud Workshops (FiCloudW). pp. 310--314. IEEE (2016)

\bibitem{halappanavar2012approximate}
Halappanavar, M., Feo, J., Villa, O., Tumeo, A., Pothen, A.: Approximate weighted matching on emerging manycore and multithreaded architectures. J. on HPC (JoHPC)  \textbf{26}(4),  413--430 (2012)

\bibitem{henzinger2020dynamic}
Henzinger, M., Khan, S., Paul, R., Schulz, C.: {Dynamic Matching Algorithms in Practice}. In: 28th Europ. Symp. on Algorithms (ESA). LIPIcs, vol.~173, pp. 58:1--58:20. Schloss Dagstuhl -- Leibniz-Zentrum f{\"u}r Informatik (2020)

\bibitem{HOEFER201320}
Hoefer, M.: Local matching dynamics in social networks. Information and Computation  \textbf{222},  20--35 (2013), intl. Coll. on Automata, Languages and Programming (ICALP)

\bibitem{huang2017approximate}
Huang, D., Pettie, S.: Approximate generalized matching: {$f$}-factors and {$f$}-edge covers. Algorithmica  \textbf{84}(7) (2022)

\bibitem{khan2016efficient}
Khan, A., Pothen, A., Mostofa Ali~Patwary, M., Satish, N.R., Sundaram, N., Manne, F., Halappanavar, M., Dubey, P.: Efficient approximation algorithms for weighted b-matching. J. on Sci. Comp. (JoSC)  \textbf{38}(5),  S593--S619 (2016)

\bibitem{khan2016designing}
Khan, A., Pothen, A., Patwary, M.M.A., Halappanavar, M., Satish, N.R., Sundaram, N., Dubey, P.: Designing scalable b-matching algorithms on distributed memory multiprocessors by approximation. In: Supercomputing (SC'16). pp. 773--783. IEEE (2016)

\bibitem{khorasani2015scalable}
Khorasani, F., Gupta, R., Bhuyan, L.N.: Scalable simd-efficient graph processing on gpus. In: Parallel Architectures and Compilation Techniques (PACT). pp. 39--50. IEEE (2015)

\bibitem{kunegis2013konect}
Kunegis, J.: Konect: the koblenz network collection. In: World Wide Web Conference (WWW). pp. 1343--1350 (2013)

\bibitem{lebedev2007using}
Lebedev, D., Mathieu, F., Viennot, L., Gai, A.T., Reynier, J., De~Montgolfier, F.: On using matching theory to understand p2p network design. In: Intl. Network Opt. Conf. (INOC) (2007)

\bibitem{liu2023tracking}
Liu, R., Huang, J., Zhang, Z.: Tracking disclosure change trajectories for financial fraud detection. Production and Operations Management  \textbf{32}(2),  584--602 (2023)

\bibitem{von2016generating}
von Looz, M., {\"O}zdayi, M.S., Laue, S., Meyerhenke, H.: Generating massive complex networks with hyperbolic geometry faster in practice. In: High Performance and Embd. Computing (HPEC). pp.~1--6. IEEE (2016)

\bibitem{manne2014new}
Manne, F., Halappanavar, M.: New effective multithreaded matching algorithms. In: Intl. Parallel and Distrib. Processing Symp. (IPDPS). pp. 519--528. IEEE (2014)

\bibitem{mestre2006greedy}
Mestre, J.: Greedy in approximation algorithms. In: European Symposium on Algorithms (ESA). pp. 528--539. Springer (2006)

\bibitem{naim2018scalable}
Naim, M., Manne, F.: Scalable b-matching on gpus. In: Intl. Parallel and Distrib. Processing Symp. Workshop (IPDPSW). pp. 637--646. IEEE (2018)

\bibitem{perumal2021electric}
Perumal, S.S., Lusby, R.M., Larsen, J.: Electric bus planning \& scheduling: A review of related problems and methodologies. European J. of Op. Research  (2021)

\bibitem{pothen2019approximation}
Pothen, A., Ferdous, S., Manne, F.: Approximation algorithms in combinatorial scientific computing. Acta Numerica  \textbf{28},  541--633 (2019)

\bibitem{samineni2024approximate}
Samineni, B., Ferdous, S., Halappanavar, M., Krishnamoorthy, B.: Approximate bipartite $ b $-matching using multiplicative auction. arXiv:2403.05781  (2024)

\bibitem{ting2015near}
Ting, H.F., Xiang, X.: Near optimal algorithms for online maximum edge-weighted b-matching and two-sided vertex-weighted b-matching. Theoretical Computer Science (TCS)  \textbf{607},  247--256 (2015)

\bibitem{wang2019adaptive}
Wang, Y., Tong, Y., Long, C., Xu, P., Xu, K., Lv, W.: Adaptive dynamic bipartite graph matching: A reinforcement learning approach. In: Intl. Conf. on Data Engineering (ICDE). pp. 1478--1489. IEEE (2019)

\end{thebibliography}

\newpage

\appendix

\section*{Appendix}

\section{Details Concerning \staticb}
\label{sub:appendix-staticb}

The original publication of \staticb\ by Khan \etal\cite{khan2016efficient} include two versions of the algorithm, a recursive one
and an iterative one. Both versions use an additional auxiliary data structure $T$, which tracks the vertices that a certain node
$v \in V$ has already been proposed to as a suitor.
Later publications, among them a GPU-parallel version of \staticb~\cite{naim2018scalable}, remove $T$ and only work with $S$ in an iterative fashion.
We use a very similar variant, shown in Algorithm~\ref{alg:static_suitor_main}.

\begin{algorithm}[h!]
    \begin{algorithmic}[1]
        \begin{scriptsize}
            \Procedure{staticBSuitor}{$G$, $b$}
            \State \textbf{Input:} Graph $G=(V,E, \omega)$, $b \in \mathbb{N}_{\geq 1}$
            \State \textbf{Output:} A $(1/2)$-approximate $b$-matching $M$
            \State $Q \gets b$ copies of $v$  $\forall~v \in V$
            \While {$Q \neq \emptyset$} \label{algline:static_loop_start}
                \State $u \gets Q.pop()$ \label{line:b_u_loop}
                        \State $p \gets \underset{v \in N(u) \setminus S(u)}{\argmax} \{ \omega(u,v) : \omega(u,v) > \omega(v, S(v).min) \}$ \label{alg_line:static_b_suitor_pot_suitor}
                        \If{$p \neq \texttt{None}$}
                            \State $y \gets S(p).insert(u)$ \Comment{Update Suitor list of $p$, removing the min node $y$} \label{algline:static_insert_u_in_p}
                            \If {$y \neq \texttt{None}$}
                                \State $Q.push(y)$ \Comment{Add the previous min entry of $S(p)$ to the queue again.} \label{algline:static_add_y_to_q}
                            \EndIf
                        \EndIf 

            \EndWhile \label{algline:static_loop_end}
            \State $M \gets \{\{u,v \} \in E \text{ s.t } v \in S(u) \wedge u \in S(v)\}$ \label{algline:static_matching_final}
            \EndProcedure
        \end{scriptsize}
    \end{algorithmic}
    \caption{Sequential $b$-Suitor algorithm}
    \label{alg:static_suitor_main}
\end{algorithm}

\begin{figure}
    \begin{center}
        \begin{tabular}{ll}
            \resizebox{\staticSize}{!}{\begin{tikzpicture}[scale=0.9]
    \tikzstyle{node}=[draw,circle,inner sep=1pt]
    \tikzstyle{active}=[draw,circle,inner sep=1pt,ultra thick,red]
    \tikzstyle{legend}=[minimum width=2.5cm,minimum height=1cm]



    \node[] (step1) at (-1,0.5) {\small{[\textbf{after $v_0$}]}};
    \node[active] (v01) at (0,-1) {$v_0$};
    \node[node] (v11) at (1.5,0) {$v_1$};
    \node[above= of v11, yshift=-3em] {\tiny{$S$: $\{v_0\}$}};
    \node[node] (v21) at (1.5,-2) {$v_2$};
    \node[below= of v21, yshift=3em] {\tiny{$S$: $\{v_0\}$}};
    \node[node] (v31) at (3.5,0) {$v_3$};
    \node[node] (v41) at (5.5,0) {$v_4$};
    \node[node] (v51) at (5.5,-2) {$v_5$};

    \draw[thick, black] (v41) edge node[font=\small, right, black] {6} (v51);
    \path[thick, black] (v41) edge node[font=\small, above, black] {9} (v31);
    \path[thick, black] (v51) edge node[font=\small, below, black] {4} (v31);
    \path[thick, black] (v31) edge node[font=\small, above, black] {1} (v11);
    \path[thick, black] (v31) edge node[font=\small, below, black] {2} (v21);
    \path[thick, black] (v21) edge node[font=\small, left, black] {1} (v11);
    \draw[thick, black] (v01) -- ++(0.75,0.5) edge[ultra thick, blue] ++(0.55,0.35) node[font=\small, above, black] {1} (v11);
    \draw[thick, black] (v01) -- ++(0.75,-0.5) edge[ultra thick, blue] ++(0.55,-0.35) node[font=\small, below, black] {1} (v21);

\end{tikzpicture}} & \resizebox{\staticSize}{!}{\begin{tikzpicture}[scale=0.9]
    \tikzstyle{node}=[draw,circle,inner sep=1pt]
    \tikzstyle{active}=[draw,circle,inner sep=1pt,ultra thick,red]
    \tikzstyle{legend}=[minimum width=2.5cm,minimum height=1cm]

    \node[] (step1) at (-1,-4) {\small{[\textbf{after $v_1$}]}};
    \node[node] (v02) at (0,-5.5) {$v_0$};
    \node[left= of v02, xshift=3em, align=left] {\tiny{$S$: $\{v_1\}$}};
    \node[active] (v12) at (1.5,-4.5) {$v_1$};
    \node[above= of v12, yshift=-3em, align=left] {\tiny{$S$: $\{v_0\}$}};
    \node[node] (v22) at (1.5,-6.5) {$v_2$};
    \node[below= of v22, yshift=3em] {\tiny{$S$: $\{v_1, v_0\}$}};
    \node[node] (v32) at (3.5,-4.5) {$v_3$};
    \node[node] (v42) at (5.5,-4.5) {$v_4$};
    \node[node] (v52) at (5.5,-6.5) {$v_5$};

    \draw[thick, black] (v42) edge node[font=\small, right, black] {6} (v52);
    \path[thick, black] (v42) edge node[font=\small, above, black] {9} (v32);
    \path[thick, black] (v52) edge node[font=\small, below, black] {4} (v32);
    \path[thick, black] (v32) edge node[font=\small, above, black] {1} (v12);
    \path[thick, black] (v32) edge node[font=\small, below, black] {2} (v22);
    \draw[thick, black] (v12) -- ++(0.0,-1.0) edge[ultra thick, blue] ++(0.0,-0.75) node[font=\small, left, black] {1} (v22);
    \draw[ultra thick, blue] (v02) edge node[font=\small, above, black] {1} (v12);
    \draw[thick, black] (v02) -- ++(0.75,-0.5) edge[ultra thick, blue] ++(0.55,-0.35) node[font=\small, below, black] {1} (v22);
\end{tikzpicture}} \\
            \resizebox{\staticSize}{!}{\begin{tikzpicture}[scale=0.9]
    \tikzstyle{node}=[draw,circle,inner sep=1pt]
    \tikzstyle{active}=[draw,circle,inner sep=1pt,ultra thick,red]
    \tikzstyle{legend}=[minimum width=2.5cm,minimum height=1cm]

    \node[] (step3) at (-1,-8.75) {\small{[\textbf{after $v_2$}]}};
    \node[node] (v03) at (0,-10.25) {$v_0$};
    \node[left= of v03, xshift=3em, align=left] {\tiny{$S$: $\{v_2,v_1\}$}};
    \node[node] (v13) at (1.5,-9.25) {$v_1$};
    \node[above= of v13, yshift=-3em, align=left] {\tiny{$S$: $\{v_0\}$}};
    \node[active] (v23) at (1.5,-11.25) {$v_2$};
    \node[below= of v23, yshift=3em, align=left] {\tiny{$S$: $\{v_1, v_0\}$}};
    \node[node] (v33) at (3.5,-9.25) {$v_3$};
    \node[above= of v33, yshift=-3em, align=left] {\tiny{$S$: $\{v_2\}$}};
    \node[node] (v43) at (5.5,-9.25) {$v_4$};
    \node[node] (v53) at (5.5,-11.25) {$v_5$};

    \draw[thick, black] (v43) edge node[font=\small, right, black] {6} (v53);
    \path[thick, black] (v43) edge node[font=\small, above, black] {9} (v33);
    \path[thick, black] (v53) edge node[font=\small, below, black] {4} (v33);
    \path[thick, black] (v33) edge node[font=\small, above, black] {1} (v13);
    \draw[thick, black] (v23) -- ++(1,1) edge[ultra thick, blue] ++(0.82,0.82) node[font=\small, below, black] {2} (v33);
    \draw[thick, black] (v13) -- ++(0.0,-1.0) edge[ultra thick, blue] ++(0.0,-0.75) node[font=\small, left, black] {1} (v23);
    \draw[ultra thick, blue] (v03) edge node[font=\small, above, black] {1} (v13);
    \draw[ultra thick, blue] (v03) edge node[font=\small, below, black] {1} (v23);
\end{tikzpicture}} & \resizebox{6.1cm}{!}{\begin{tikzpicture}[scale=0.9]
    \tikzstyle{node}=[draw,circle,inner sep=1pt]
    \tikzstyle{active}=[draw,circle,inner sep=1pt,ultra thick,red]
    \tikzstyle{legend}=[minimum width=2.5cm,minimum height=1cm]

    \node[] (step1) at (6.5,0.5) {\small[{\textbf{after $v_3$}]}};
    \node[node] (v04) at (7.5,-1) {$v_0$};
    \node[left= of v04, xshift=3em, align=left] {\tiny{$S$: $\{v_2,v_1\}$}};
    \node[node] (v14) at (9,0) {$v_1$};
    \node[above= of v14, yshift=-3em, align=left] {\tiny{$S$: $\{v_0\}$}};
    \node[node] (v24) at (9,-2) {$v_2$};
    \node[below= of v24, yshift=3em, align=left] {\tiny{$S$: $\{v_1, v_0\}$}};
    \node[active] (v34) at (11,0) {$v_3$};
    \node[above= of v34, yshift=-3em, align=left]  {\tiny{$S$: $\{v_2\}$}};
    \node[node] (v44) at (13,0) {$v_4$};
    \node[above= of v44, yshift=-3em, align=left] {\tiny{$S$: $\{v_3\}$}};
    \node[node] (v54) at (13,-2) {$v_5$};
    \node[below= of v54, yshift=3em, align=left] {\tiny{$S$: $\{v_3\}$}};

    \draw[thick, black] (v44) edge node[font=\small, right, black] {6} (v54);
    \draw[thick, black] (v34) -- ++(1,0) edge[ultra thick, blue] ++(0.75,0) node[font=\small, above, black] {9} (v44);
    \draw[thick, black] (v34) -- ++(1,-1) edge[ultra thick, blue] ++(0.82,-0.82) node[font=\small, below, black] {4} (v54);
    \path[thick, black] (v34) edge node[font=\small, above, black] {1} (v14);
    \draw[thick, black] (v24) -- ++(1,1) edge[ultra thick, blue] ++(0.82,0.82) node[font=\small, below, black] {2} (v34);
    \draw[thick, black] (v14) -- ++(0.0,-1.0) edge[ultra thick, blue] ++(0.0,-0.75) node[font=\small, left, black] {1} (v24);
    \draw[ultra thick, blue] (v04) edge node[font=\small, above, black] {1} (v14);
    \draw[ultra thick, blue] (v04) edge node[font=\small, below, black] {1} (v24);

\end{tikzpicture}} \\
            \resizebox{6.3cm}{!}{\begin{tikzpicture}[scale=0.9]
    \tikzstyle{node}=[draw,circle,inner sep=1pt]
    \tikzstyle{active}=[draw,circle,inner sep=1pt,ultra thick,red]
    \tikzstyle{legend}=[minimum width=2.5cm,minimum height=1cm]

    \node[] (step5) at (6.5,-4.0) {\small{[\textbf{after $v_4$}]}};
    \node[node] (v05) at (7.5,-5.5) {$v_0$};
    \node[left= of v05, xshift=3em, align=left] {\tiny{$S$: $\{v_2,v_1\}$}};
    \node[node] (v15) at (9,-4.5) {$v_1$};
    \node[above= of v15, yshift=-3em, align=left] {\tiny{$S$: $\{v_0\}$}};
    \node[node] (v25) at (9,-6.5) {$v_2$};
    \node[below= of v25, yshift=3em, align=left] {\tiny{$S$: $\{v_1, v_0\}$}};
    \node[node] (v35) at (11,-4.5) {$v_3$};
    \node[above= of v35, yshift=-3em, align=left]  {\tiny{$S$: $\{v_2, v_4\}$}};
    \node[active] (v45) at (13,-4.5) {$v_4$};
    \node[above= of v45, yshift=-3em, align=left]  {\tiny{$S$: $\{v_3\}$}};
    \node[node] (v55) at (13,-6.5) {$v_5$};
    \node[below= of v55, yshift=3em, align=left] {\tiny{$S$: $\{v_3, v_4\}$}};

    \draw[thick, black] (v45) -- ++(0,-1) edge[ultra thick, blue] ++(0.0,-0.82) node[font=\small, right, black] {6} (v55);
    \draw[ultra thick, blue] (v35) edge node[font=\small, above, black] {9} (v45);
    \draw[thick, black] (v35) -- ++(1,-1) edge[ultra thick, blue] ++(0.82,-0.82) node[font=\small, below, black] {4} (v55);
    \path[thick, black] (v35) edge node[font=\small, above, black] {1} (v15);
    \draw[thick, black] (v25) -- ++(1,1) edge[ultra thick, blue] ++(0.82,0.82) node[font=\small, below, black] {2} (v35);
    \draw[thick, black] (v15) -- ++(0.0,-1.0) edge[ultra thick, blue] ++(0.0,-0.75) node[font=\small, left, black] {1} (v25);
    \draw[ultra thick, blue] (v05) edge node[font=\small, above, black] {1} (v15);
    \draw[ultra thick, blue] (v05) edge node[font=\small, below, black] {1} (v25);
\end{tikzpicture}} & \resizebox{6.2cm}{!}{\begin{tikzpicture}[scale=0.9]
    \tikzstyle{node}=[draw,circle,inner sep=1pt]
    \tikzstyle{active}=[draw,circle,inner sep=1pt,ultra thick,red]
    \tikzstyle{legend}=[minimum width=2.5cm,minimum height=1cm]

    \node[] (step5) at (6.5,-8.75) {\small{[\textbf{after $v_5$}]}};
    \node[node] (v06) at (7.5,-10.25) {$v_0$};
    \node[left= of v06, xshift=3em, align=left] {\tiny{$S$: $\{v_2,v_1\}$}};
    \node[node] (v16) at (9,-9.25) {$v_1$};
    \node[above= of v16, yshift=-3em, align=left] {\tiny{$S$: $\{v_0\}$}};
    \node[node] (v26) at (9,-11.25) {$v_2$};
    \node[below= of v26, yshift=3em, align=left] {\tiny{$S$: $\{v_1, v_0\}$}};
    \node[node] (v36) at (11,-9.25) {$v_3$};
    \node[above= of v36, yshift=-3em, align=left] {\tiny{$S$: $\{\xcancel{v_2}, v_5, v_4\}$}};
    \node[node] (v46) at (13,-9.25) {$v_4$};
    \node[above= of v46, yshift=-3em, align=left] {\tiny{$S$: $\{v_3,v_5\}$}};
    \node[active] (v56) at (13,-11.25) {$v_5$};
    \node[below= of v56, yshift=3em, align=left] {\tiny{$S$: $\{v_3,v_4\}$}};

    \draw[ultra thick, blue] (v46) edge node[font=\small, right, black] {6} (v56);
    \draw[ultra thick, blue] (v36) edge node[font=\small, above, black] {9} (v46);
    \draw[ultra thick, blue] (v36) edge node[font=\small, below, black] {4} (v56);
    \path[thick, black] (v36) edge node[font=\small, above, black] {1} (v16);
    \path[thick, black] (v26) edge node[font=\small, below, black] {2} (v36);
    \draw[thick, black] (v16) -- ++(0.0,-1.0) edge[ultra thick, blue] ++(0.0,-0.75) node[font=\small, left, black] {1} (v26);
    \draw[ultra thick, blue] (v06) edge node[font=\small, above, black] {1} (v16);
    \draw[ultra thick, blue] (v06) edge node[font=\small, below, black] {1} (v26);
\end{tikzpicture}}       \\
            \resizebox{6.3cm}{!}{\begin{tikzpicture}[scale=0.9]
    \tikzstyle{node}=[draw,circle,inner sep=1pt]
    \tikzstyle{active}=[draw,circle,inner sep=1pt,ultra thick,red]
    \tikzstyle{legend}=[minimum width=2.5cm,minimum height=1cm]

    \node[] (step5) at (2.25,-13.5) {\small[{\textbf{final $v_2$}]}};
    \node[node] (v07) at (3.25,-15) {$v_0$};
    \node[left= of v07, xshift=3em, align=left] {\tiny{$S$: $\{v_2, v_1\}$}};
    \node[node] (v17) at (4.75,-14) {$v_1$};
    \node[above= of v17, yshift=-3em, align=left] {\tiny{$S$: $\{v_2, v_0\}$}};
    \node[active] (v27) at (4.75,-16) {$v_2$};
    \node[below= of v27, yshift=3em, align=left] {\tiny{$S$: $\{v_1, v_0\}$}};
    \node[node] (v37) at (6.75,-14) {$v_3$};
    \node[above= of v37, yshift=-3em, align=left] {\tiny{$S$: $\{v_5, v_4\}$}};
    \node[node] (v47) at (8.75,-14) {$v_4$};
    \node[above= of v47, yshift=-3em, align=left] {\tiny{$S$: $\{v_3, v_5\}$}};
    \node[node] (v57) at (8.75,-16) {$v_5$};
    \node[below= of v57, yshift=3em, align=left] {\tiny{$S$: $\{v_3, v_4\}$}};

    \draw[ultra thick, blue] (v47) edge node[font=\small, right, black] {6} (v57);
    \draw[ultra thick, blue] (v37) edge node[font=\small, above, black] {9} (v47);
    \draw[ultra thick, blue] (v37) edge node[font=\small, below, black] {4} (v57);
    \path[thick, black] (v37) edge node[font=\small, above, black] {1} (v17);
    \path[thick, black] (v27) edge node[font=\small, below, black] {2} (v37);
    \draw[ultra thick, blue] (v17) edge node[font=\small, left, black] {1} (v27);
    \draw[ultra thick, blue] (v07) edge node[font=\small, above, black] {1} (v17);
    \draw[ultra thick, blue] (v07) edge node[font=\small, below, black] {1} (v27);
\end{tikzpicture}} & \resizebox{\staticSize}{!}{\begin{tikzpicture}[scale=0.9]
    \tikzstyle{node}=[draw,circle,inner sep=1pt,gray]
    \tikzstyle{active}=[draw,circle,inner sep=1pt,ultra thick,red]
    \tikzstyle{legend}=[minimum width=2.5cm,minimum height=1cm]

    \node (blank) at (-2,3) {};

    \node[legend] (legend) at (1,1) {};
    \node[above= of legend,yshift=-3em,xshift=-1.8em,align=left] {\tiny{\textbf{Legend:}}};
    \draw[thick, black] (0,0.7) edge node[font=\small, right, black] {\tiny{ unmatched edge}} (0.4,0.7);
    \draw[thick, black] (0,1.0) -- ++(0.2,0) edge[ultra thick, blue] ++(0.2,0.0) node[font=\small, right, black] {\tiny{ half-matched edge}} (0.4,1.0);
    \draw[ultra thick, blue] (0,1.3) edge node[font=\small, right, black] {\tiny{ matched edge}} (0.4,1.3);
    \node[] (order) at (4,2.75) {\tiny{\textbf{Order:}}};
    \node[node] (step1) at (3.5,2.25) {\tiny{0}};
    \node[node] (step2) at (4.5,2.25) {\tiny{1}};
    \node[node] (step3) at (3.5,1.75) {\tiny{2}};
    \node[node] (step4) at (4.5,1.75) {\tiny{3}};
    \node[node] (step5) at (3.5,1.25) {\tiny{4}};
    \node[node] (step6) at (4.5,1.25) {\tiny{5}};
    \node[node] (step7) at (3.5,0.75) {\tiny{6}};

    \draw [->,gray] (step1) edge (step2) (step2) edge (step3) (step3) edge (step4) (step4) edge (step5) (step5) edge (step6) (step6) edge (step7);


\end{tikzpicture}} \\
        \end{tabular}
    \end{center}
    \caption{Process of the sequential \staticb\ $(\frac{1}{2})$-approximation algorithm by Khan \etal\cite{khan2016efficient}.
    Each subfigure shows the state after one node is handled by Algorithm \ref{alg:static_suitor_main}.
    For [\textbf{after $v_0$}], $\dots$, [\textbf{after $v_4$}], only new suitors are added to their respective priority queues $S$. In step [\textbf{after $v_5$}], a previous suitor from $v_3$ is removed, leading to an additional recursive step [\textbf{final $v_2$}] for $v_2$ in order to find an alternative suitor.
    For this instance, the algorithm produces the maximum weight $2$-matching.}
        \label{fig:static_suitor}
\end{figure}
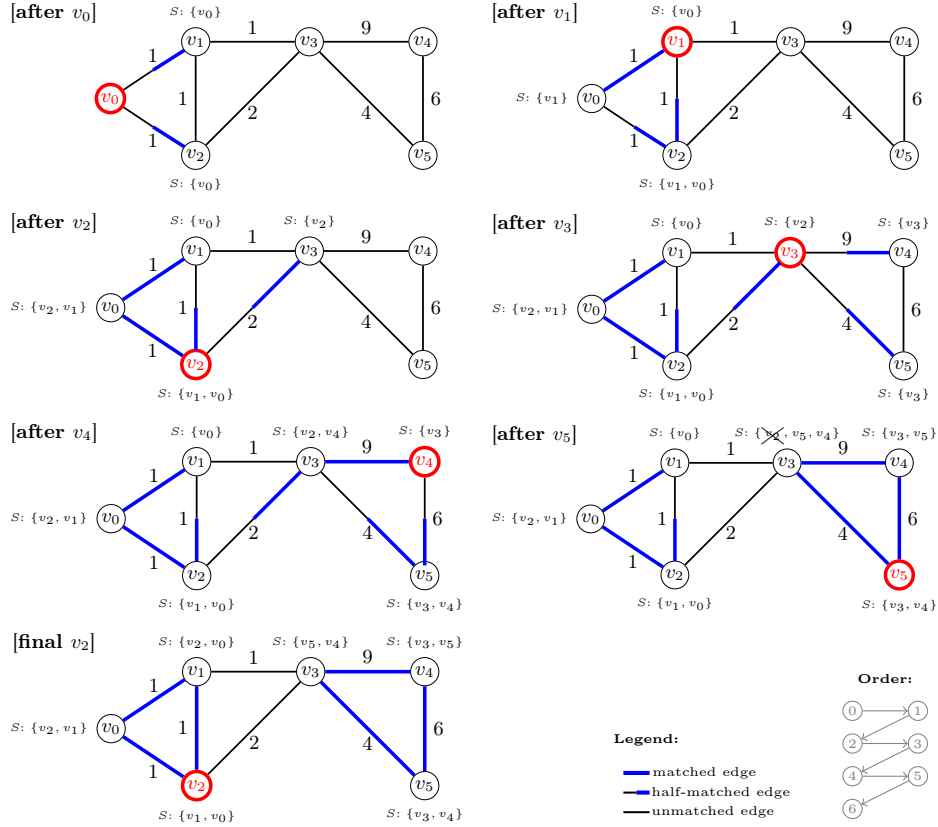

An example of how a $2$-matching is formed using Algorithm~\ref{alg:static_suitor_main} is shown in Figure~\ref{fig:static_suitor}.
Note that a \textcolor{blue}{blue} half-edge denotes that $u$ is part of $S(v)$ but not vice versa. A \textcolor{blue}{blue} edge denotes $v \in S(u) \Leftrightarrow u \in S(v)$.
During the matching of $v_5$, $v_3$ gets a better proposal $\omega(v_5, v_3) > \omega(v_3, v_2)$, therefore the min suitor $v_2$ is deleted from the suitor list of $v_3$.
After the final recursive step (finding a new possible suitor for $v_2$), the algorithm outputs a valid $b$-matching (highlighted in \textcolor{blue}{blue}).
After the termination of \staticb, the invariant $v \in S(u) \Leftrightarrow u \in S(v)$ holds.
(For \dynb\ this invariant is also maintained during the update phase of the algorithm.)

\newpage

\section{Proving that \dynb\ and \staticb\ Yield the Same Result}
\label{sub:appx-equality-proof}
\subsection{Edge Insertion}
\label{sub:appendix-edge-insertion}

Let us consider the case that a new edge $\{u,v\} \notin E$ is added to the graph.
This edge influences the matching only if its edge weight exceeds the edge weights $\omega\left(u,S^{(i)}(u).min\right)$ and $\omega\left(v,S^{(i)}(v).min\right)$ (note: any edge weight is better than \texttt{None}). This property was shown in Ref.~\cite{angriman2022batch} for $1$-matchings; we extend this result to $b$-matchings
with a completely new proof.

\begin{lemma}
    Let $\widetilde{G}=(V,\widetilde{E})=(V,E\cup\{e\})$ with $e=\{u,v\} \notin E$. Then: $e \in \widetilde{M} \iff \omega(u,v) > \max\{\omega(u,S(u).min), \omega(v,S(v).min)\}$.
    \label{lemma:new_edge_def_1_new}
\end{lemma}

\begin{proof}
    \label{proof:lemma_edge_insertion_1}
    "$\leftarrow$": We assume $\omega(u,v) > \max\{\omega(u,S(u).min), \omega(v,S(v).min)\}~(\star)$.
    To prove our claim, we need to show that $\{u,v\} \in \widetilde{M}$, \ie the static Suitor algorithm inserts \{u,v\} into the original matching $M$.
    Thus, when iterating over $u$ (lines~\ref{algline:static_loop_start} - \ref{algline:static_loop_end} of Alg.~\ref{alg:static_suitor_main}) and due to $(\star)$, there must be an iteration \drop{$i \leq b$} in which $v$ is assigned to $p$ in line~\ref{alg_line:static_b_suitor_pot_suitor}. 
    In line~\ref{algline:static_insert_u_in_p}, $u$ is inserted into $\widetilde{S}(v)$ (since $p=v$).
    The same happens vice versa in another iteration for $v$, where $u$ is assigned to $p$ and $v$ is inserted in  $\widetilde{S}(u)$.
    Now that $v \in \widetilde{S}(u)$ and $u \in \widetilde{S}(v)$, we still have to show that they are never evicted.
    We consider $v$, the argument is symmetric for $u$.
    
    Let us assume for sake of contradiction that there are at least $b$ neighbors of $v$ that insert itself into $\widetilde{S}(v)$ and evict $u$.
    Let us denote this set by $C$. Then, for each $x \in C:~\omega(v,x)>\omega(v,u)$.
    Out of the elements in $C$, the top-$b$ would have been part of the original matching $M$, otherwise they would not be part of $\widetilde{S}(v)$. However, due to $(\star)$, $\omega(u,v)$ is larger than $b-1$ of the top-$b$ entries in $C$, leading to a contradiction concerning the eviction of $u$ of $\widetilde{S}(v)$.

    "$\rightarrow$": We assume $e=\{u,v\} \in \widetilde{M}$.
    There exist three cases for $\omega(u,v)$ with respect to $\omega(u,S(u).min)$ and $\omega(v,S(v).min)$: \textbf{(i)} $\omega(u,v)$ is smaller than both, \textbf{(ii)} $\omega(u,v)$ is smaller than one and larger than the other, \textbf{(iii)} $\omega(u,v)$ is larger than both.
    We show that cases \textbf{(i)} and \textbf{(ii)} are not possible if $\{u,v\} \in \widetilde{M}$.

    For case \textbf{(i)}, since their respective $min$-values cannot be \texttt{None}, $u$ and $v$ are both saturated in $M$.
    As a result, both edges $\{u, S(u).min\}$ and $\{v, S(v).min\}$ are part of the original matching $M$.
    If $\{u,v\} \in \widetilde{M}$, that means the two edges $\{u, S(u).min\}$ and $\{v, S(v).min\}$ are evicted from $\widetilde{M}$ due to saturation.
    But the static suitor algorithm would have chosen $\widetilde{S}(v).min$ over $u$ $\left[\widetilde{S}(u).min \text{ over } v\right]$ due to its higher weight, leading to a contradiction. 

    For case \textbf{(ii)}, we assume w.l.o.g. that $\omega(u,v) < \omega(u,S(u).min)$.
    If $u$ gets processed in Algorithm~\ref{alg:static_suitor_main} before $v$, then $u$ becomes saturated in this process due to $S(u).min$ being not \texttt{None}.
    Also, $u \notin \widetilde{S}(v)$ due to our assumption.
    While processing $v$, line~\ref{alg_line:static_b_suitor_pot_suitor} cannot return $u$, since $\{u,v\}$ is not a better proposal.
    If $v$ gets processed before $u$, then the new edge might be regarded, \ie $v \in \widetilde{S}(u)$.
    However, there exist at least $b$ other vertices which insert themselves to $\widetilde{S}(u)$, eventually evicting $v$.
    Finally, due to line~\ref{algline:static_add_y_to_q}, $v$ is inserted into $Q$ again.
    But then, line~\ref{alg_line:static_b_suitor_pot_suitor} cannot return $u$, since $\{u,v\}$ is not a better proposal.
    Thus, we have shown that edge $e=\{u,v\} \notin \widetilde{M}$ for \textbf{(i)} and \textbf{(ii)}. 
    
    Case \textbf{(iii)} is now the only remaining option for $e=\{u,v\}$. According to $(\star)$, $u$ is added to $\widetilde{S}(v)$ when $u$ is processed. 
    The same happens vice versa for $v$.
    Over the course of the main loop, there are at most $b-1$ other insertions into $\widetilde{S}(u)$ and $\widetilde{S}(v)$, not evicting $v$ (or $u$, respectively) again.
    Finally, in line~\ref{algline:static_matching_final}, $\widetilde{M}$ is set such that $e=\{u,v\}$ is included.
\qed
\end{proof}

\begin{algorithm}[tb]
    \begin{algorithmic}[1]
        \begin{scriptsize}
            \Procedure{EdgeInsertion}{$\widetilde{G}$, $\{u, v\}$}
            \State \textbf{Input:} Graph $\widetilde{G}=(V,\widetilde{E})$, an edge  $\{u, v\} \notin E$
            \State \textbf{Output:} New $(1/2)$-approximate b-Matching $M^{(f)}$
            \vspace{0.2cm}
            \State $M^{(i)} \gets \{\{u^\prime,v^\prime \} \in E \text{ s.t } v^\prime \in S(u^\prime) \wedge u^\prime \in S(v^\prime)\}$ \label{algline:edge_insertion_m_init} 
            \If{$\omega(u,v) > \max\{\omega(u,S(u).min), \omega(v,S(v).min)\}$} 
            \State $startU \gets S^{(i)}(u).insert(v)$ \Comment{Update suitor queue of $u$, returns replaced item} \label{algline:push_uv}
            \State $startV \gets S^{(i)}(v).insert(u)$ \label{algline:push_vu}
            \If{ $startU \neq \texttt{None}$}
            \State $S^{(i)}(startU).remove(u)$ \Comment{Maintain \textit{S-invariant}} \label{algline:edge_insertion_s_invariant}
            \State $\textsc{FindAffected}(startU)$ \Comment{Check for new suitors.}
            \EndIf
            \If{ $startV \neq \texttt{None}$}
            \State $S^{(i)}(startV).remove(v)$
            \State $\textsc{FindAffected}(startV)$
            \EndIf
            \State $M^{(i)} \gets \{\{u^\prime,v^\prime \} \in \widetilde{E} \text{ s.t } v^\prime \in S^{(i)}(u^\prime) \wedge u^\prime \in S^{(i)}(v^\prime)\}$ \label{algline:edge_insertion_mi}
            \EndIf
            \State $M^{(f)} \gets M^{(i)}$
            \EndProcedure
        \end{scriptsize}
    \end{algorithmic}
    \caption{Dynamic Edge Insertion}
    \label{alg:dynamic_edge_insertion_new}
\end{algorithm}

Using Lemma \ref{lemma:new_edge_def_1_new}, we design Algorithm \ref{alg:dynamic_edge_insertion_new} for the already altered graph $\widetilde{G}$. If according to Lemma \ref{lemma:new_edge_def_1_new} the new edge creates a better suitor for $u$ and $v$, the suitor lists of both endpoints $u$ and $v$ have to be updated.
If in this process a previous suitor is removed from $S^{(i)}(u)$ $[S^{(i)}(v)]$ (line~\ref{algline:push_uv}), the {S-invariant} is directly maintained by removing $u$ $[v]$ from the respective suitor lists.
In addition, we call \findaff\ for the removed nodes, as they might be able to find new partners (\ie they are {locally affected} according to Def.~\ref{def:affected}). We refer to these nodes as $startU$ and $startV$. By calling \findaff\ twice (once for $startU$, once for $startV$), we identify and update all \laff\ vertices influenced by inserting $e=\{u,v\}$. For proving the properties of our \dynb\ algorithm, it is useful to define the notion of an \uppath.

\begin{definition}
The vertices visited and updated during the execution of \\ \findaff\ form a path which we call \uppath. 
An update path $P_x$ (of length $q \geq 0$) starts with vertex $p^{(x)}_0 = x$ and ends with some vertex $p^{(x)}_q$.
\label{definition:update_path}
\end{definition}

If $q=0$, then $P_x$ consists only of $p^{(x)}_0$ and we refer to it as an \textit{empty} \uppath.
We denote every node in the path as $p^{(x)}_j$, $0 \leq j \leq q$, with $j$ being its position in the path.
Each node pair $(p^{(x)}_j, p^{(x)}_{j+1})$ relates to an edge $e^{(x)}_j=\{p^{(x)}_j, p^{(x)}_{j+1}\}$ which is also present in $\widetilde{G}$.

When calling \findaff, we process an update path, changing intermediate suitor lists $S^{(i)}$ 
of all nodes on the (non-empty) path.
We continue by proving several properties of an \uppath\ as they arise from the program flow of \findaff. The following lemma shows that the edge weights decrease when traversing an update path from its start to its end.

\begin{lemma}
    Let $P_x$ be an update path traversed by \findaff. For two edges $e_i \neq e_j \in P_x$, it holds that
    $\omega(e_i) < \omega(e_j)$ [$\omega(e_i) > \omega(e_j)$] if $i>j$ [$i<j$]. 
    \label{lemma:path_property_weight}
\end{lemma}

\begin{proof}
    Let $P_x$ be an update path traversed by \findaff.
    The corresponding start of $P_x$ is $p^{(x)}_0 = x$ as given by Definition~\ref{definition:update_path}.
    The first visited node is assigned to $cu$ (line~\ref{algline:init-cu}). \findaff\ looks for a new potential partner $ca$ in its neighborhood (line~\ref{algline:check_def_affected}).
    If one is found, Definition~\ref{def:affected} is fulfilled ($x$ is a \laff\ vertex), leading to a suitor list update such that $ca \in S^{(i)}(x)$ and vice versa (lines~\ref{algline:s-invariant1} and~\ref{algline:s-invariant2}).
    We add $ca$ also to the path $P_x$, effectively setting $p^{(x)}_1 = ca$ and $e^{(x)}_0=\{x,p^{(x)}_1\}$.
    There exist two cases for $p^{(x)}_1$. If $p^{(x)}_1$ was \usat\ before adding $x$, our update path $P_x$ is finished.

    Since $e_i \neq e_j$, for our proof we assume that $P_x$ has more than one edge, \ie $p^{(x)}_1$ was \usat\ before adding $x$. 
    Assuming this situation, $x$ replaces the previous $min$-entry $prevCa$.
    This effectively adds $p^{(x)}_2= prevCa$ to $P_x$.
    According to Lemma \ref{lemma:new_edge_def_1_new} and the tie-breaking rule, $\omega(p^{(x)}_1,p^{(x)}_2) < \omega(x,p^{(x)}_1)$;
    otherwise $p^{(x)}_1$ would not have been chosen to be a potential suitor for~$x$.

    In line~\ref{algline:prevca_cu} we set $cu$ to $prevCa$, effectively repeating the main loop.
    Again, there are two cases for $p^{(x)}_2$: either it is \laff\ or not.
    In case it is not, \ie there is no potential new suitor in the neighborhood, the update path is not continued.
    If $p^{(x)}_2$ is \laff, we follow the same process as for $x$ (main loop in \findaff) and another $p^{(x)}_3$ to $P_x$ with again $\omega(p^{(x)}_2,p^{(x)}_3) < \omega(p^{(x)}_1,p^{(x)}_2)$.
    This process continues until either we find no new potential suitor $p^{(x)}_{j+1}$ for $p^{(x)}_j$ with $j$ even or $p^{(x)}_{j+1}$ is \usat. In both cases, the loop terminates.
    
    We can continue the above argument inductively. Thus, given that we add the nodes to $P_x$ in the same order as updates occur, the edges formed by successive node pairs $(p^{(x)}_{j},p^{(x)}_{j+1}) \in P_x$ have decreasing weight, \ie $\omega(e_{j})=\omega(p^{(x)}_{j},p^{(x)}_{j+1}) < \omega(p^{(x)}_{j-1},p^{(x)}_{j})=\omega(e_{j-1})$.
    \qed
\end{proof}

\begin{lemma}
    Let $P_x$ be an update path traversed by \findaff. Then $P_x$ is simple.
    \label{lemma:path_property_matching_simple}
\end{lemma}

\begin{proof}
    Due to the decreasing weight of edges in the update path (Lemma~\ref{lemma:path_property_weight}), the weight of $e^{(x)}_{j+1}$ is smaller than the weight of an edge $e^{(x)}_{j}$, \ie $\omega(p^{(x)}_{j+1},p^{(x)}_{j+2}) < \omega(p^{(x)}_{j},p^{(x)}_{j+1})$. In other words, whenever a vertex $x_j$ loses its suitor $x_{prev}=S^{(i)}(x_j).min$, it holds: if it finds a new suitor $x_k$, then $\omega(x_j, x_k)$ is smaller than $\omega(x_j, x_{prev})$. As a result, $x_j$ cannot be added to the suitor list of a vertex already present in the update path, otherwise the corresponding edge weight would have been larger than the previous min entry of that node. Hence, the update path is simple.
    \qed 
\end{proof}

Due to the construction structure of \findaff, we know that the first edge on an update path represents an addition to $M^{(i)}$.
Together with the path properties from Lemma~\ref{lemma:path_property_matching_simple} and its proof, we can derive:

\begin{lemma}
    Let $P_x$ be a path processed when calling \findaff\ and let $M^{(i)} = \{\{u^\prime,v^\prime \} \in \widetilde{E} \text{ s.t } v^\prime \in S^{(i)}(u^\prime) \wedge u^\prime \in S^{(i)}(v^\prime)\}$. If $j$ is even [odd], the operation connected to an edge $e^{(x)}_j$ is an insertion [removal] of an element into [from] $M^{(i)}$.\label{lemma:even_odd}
\end{lemma}

\begin{proof}
    Based on the construction of $P_x$ in \findaff, the first edge represents the insertion of $p^{(x)}_{0}$ into $S^{(i)}(p^{(x)}_{1})$ and 
    vice versa. For $M^{(i)} = \{\{u^\prime,v^\prime \} \in \widetilde{E} \text{ s.t } v^\prime \in S^{(i)}(u^\prime) \wedge u^\prime \in S^{(i)}(v^\prime)\}$, the algorithm inserts $e^{(x)}_0$ into $M^{(i)}$. The next edge consists of $p^{(x)}_{1}$ and its previous minimum suitor before the update resulting from processing $e^{(x)}_0$. During the main loop of \findaff, $p^{(x)}_{1}$ is \drop{directly} removed from $S^{(i)}(p^{(x)}_{2})$ and vice versa. In order for the entries to be removed symmetrically, $e^{(x)}_1$ was contained in $M^{(i)}$ before the
    update and is removed afterwards.
    
    Due to the construction of $P_x$, the above argument about $e^{(x)}_0$ and $e^{(x)}_1$ can be repeated inductively for all following edges. Therefore each edge with an even index represents an insertion of an element in $M^{(i)}$ and each edge with an odd index represents a removal of an element from $M^{(i)}$.
    \qed
\end{proof}

In the following, we discuss in more detail under which circumstances our dynamic procedure \edgeins\ calls \findaff\ and therefore processes update paths.
Similar to Angriman \etal\cite[Sec. 3.1]{angriman2022batch}, we examine three distinct cases -- addressing the saturation state of the involved endpoints of the new edge $\{u,v\}$: \textbf{(i)} both $u$ and $v$ are {unsatured} (Lemma \ref{lemma:edge_insertion_unsaturated}), \textbf{(ii)} only one of them is {saturated} (Lemma \ref{lemma:edge_insertion_one_saturated}) and \textbf{(iii)} both are {saturated} (Lemma~\ref{lemma:edge_insertion_loose_end} and \ref{lemma:edge_insertion_both_saturated}).

\subsubsection{Case \textbf{(i)}: Both $u$ and $v$ are \usat} 

If both $u$ and $v$ are \usat, we can directly update the suitor lists of both endpoints. Also, since no other vertex is influenced by this update, \findaff\ is not called. This situation can be formalized as follows:

\begin{lemma}
    Let $e := \{u,v\} \notin E$ and and $\omega(u,v) > \max\{\omega(u,S(u).min),$ \\ $\omega(v,S(v).min)\}$. If both $u$ and $v$ are {unsatured}, then $M^{(f)} = \widetilde{M} = M^{(i)} \cup\{e\} = M\cup\{e\}$ and no other vertices are locally affected.
    \label{lemma:edge_insertion_unsaturated}
\end{lemma}

\begin{proof}
    After the initial computation of $M$, for both $u$ and $v$ there is no better suitor $\argmaxA_{v \in N(u) \setminus S(u)} \{ \omega\left(u,v\right) \vert \omega\left(u,v\right) > \omega\left(v, S(v).min\right) \}$.
    Once edge $\{u,v\}$ is inserted, $u$ and $v$ become locally affected according to Definition \ref{def:affected} with ${s}(u) = v$ and vice versa.
    Since $S^{(i)}(u).min$ and $S^{(i)}(v).min$ were both $\texttt{None}$, pushing into the respective suitor lists does not remove the previous $min$-entry.
    The suitor lists of the remaining vertices $V\setminus \{u,v\}$ are therefore not influenced by the edge insertion. 
    As a result, all edges in $M$ are still part of $\widetilde{M}$, since the only change is the addition of $u$ in $S(v)$ and vice versa;
    in particular, we do not remove any items from any suitor lists.
    With Proposition~\ref{proposition:valid_matching_static}, it follows that $\widetilde{M} =  M\cup\{e\}$.

    The dynamic algorithm follows that logic by only changing the suitor lists of $u$ and $v$ in case $S^{(i)}(u).min$ and $S^{(i)}(v).min$ are $\texttt{None}$ before the update.
    That concludes the \textsc{EdgeInsertion} algorithm and the final suitor lists $S^{(f)}$ fulfill the {S-invariant}.
    $M^{(f)}$ contains $\{u,v\}$ and therefore also $\widetilde{M} = M^{(f)}$.
        \qed 
\end{proof}

As given by Lemma~\ref{lemma:edge_insertion_unsaturated}, $startU$ and $startV$ are \texttt{None} for case \textbf{(i)} and the resulting update paths are empty. Non-empty paths can only occur for case \textbf{(ii)} and \textbf{(iii)}, \ie if at least one endpoint is {saturated}.

\subsubsection{Case \textbf{(ii)}: One endpoint is \sat} 

If one endpoint of the new edge is already \sat, then an update to its suitor lists also influences the suitor list of one of its neighbors, \ie the previous minimum suitor. If this influenced vertex can find a new suitor, we know by Definition~\ref{def:affected} that it
    is \laff. We formalize this in the following lemma.

\begin{lemma}
    Let $e := \{u,v\} \notin E$ and $\omega(u,v) > \max\{\omega(u,S(u).min),$ \\ $\omega(v,S(v).min)\}$. If $u$ is {saturated} and $v$ is {unsaturated} $[$or vice versa$]$, then 
    \textsc{EdgeInsertion} traverses an update path $P_u$ $\left[P_v\right]$.
    \label{lemma:edge_insertion_one_saturated}
\end{lemma}

\begin{proof}
    Since the proof is symmetric for both cases, we assume w.l.o.g. that $v$ is the {unsaturated} vertex, while $u$ is {saturated}.
    For $v$, the situation is similar to the one in Lemma \ref{lemma:edge_insertion_unsaturated}, where \textit{startV} is \texttt{None} and therefore \textsc{FindAffected} is not called.
    Since $u$ is saturated, we initiate the update procedure by calling $\textsc{FindAffected}(startU)$,
    which traverses the \uppath.
    \qed 
\end{proof}

\subsubsection{Case \textbf{(iii)}: Both $u$ and $v$ are \sat} 
When both endpoints of the new edge are \sat, we can easily derive Corollary~\ref{corollary:edge_insertion_both_saturated} from Lemma~\ref{lemma:edge_insertion_one_saturated}. 
The only necessary change is that \textsc{EdgeInsertion}$(u,v)$ now traverses two update paths $P_u$ and $P_v$.

\begin{corollary}
    Let $e := \{u,v\} \notin E$ and $\omega(u,v) > \max\{\omega(u,S(u).min),$ \\ $\omega(v,S(v).min)\}$. If both $u$ and $v$ are {saturated}, then two update paths $P_u$ and $P_v$ are traversed; they start at $startU$ and $startV$, respectively.
    \label{corollary:edge_insertion_both_saturated}
\end{corollary}

Both $P_u$ and $P_v$ are simple paths with decreasing edge weights. Therefore, updates to nodes can only happen multiple times if there exists some overlap between $P_u$ and $P_v$. Under the circumstances described below, these overlaps can lead to a situation where an update path leads to another update path. This is handled by \findaff\ in line~\ref{algline:findaffected_loose_end}, where the starting point of this additional update path is initalized as $looseEnd$. After \findaff\ has finished identifying and updating all vertices on one update path, a recursive call handles this additional update path. In the following, we refer to this update path as $P_a$ and describe the situation in which this additional path occurs.

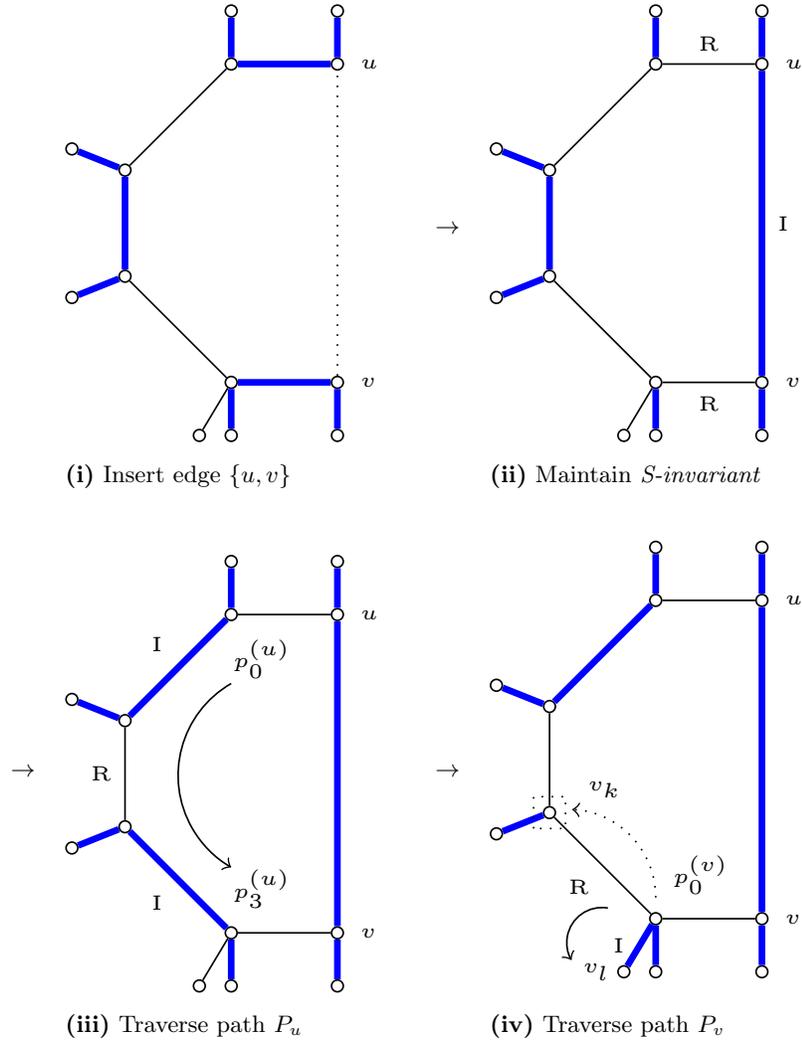
\begin{figure*}
    \begin{center}
        \begin{tblr}{Q[0.3cm,valign=m] Q[\looseEndSize, valign=h] Q[0.3cm,valign=m] Q[\looseEndSize, valign=h]}
            & \resizebox{\looseEndSize}{!}{\begin{tikzpicture}[scale=0.9]
    \tikzstyle{node}=[draw,circle,inner sep=1pt]
    \tikzstyle{active}=[draw,circle,inner sep=1pt,ultra thick,red]
    \tikzstyle{legend}=[minimum width=2.5cm,minimum height=1cm]

    \node[node] (v) at (0,0) {};
    \node[right= of v, xshift=-3em] {\tiny{$v$}};
    \node[node] (vadd) at (0,-0.5) {};

    \node[node] (startV) at (-1,0) {};
    \node[node] (startVadd1) at (-1,-0.5) {};
    \node[node] (startVadd2) at (-1.3,-0.5) {};

    \node[node] (v1) at (-2,1) {};
    \node[node] (v1add) at (-2.5,0.8) {};

    \node[node] (v2) at (-2,2) {};
    \node[node] (v2add) at (-2.5,2.2) {};

    \node[node] (startU) at (-1,3) {};
    \node[node] (startUadd) at (-1,3.5) {};
    
    \node[node] (u) at (0,3) {};
    \node[right= of u, xshift=-3em] {\tiny{$u$}};
    \node[node] (uadd) at (0, 3.5) {};

    \draw[ultra thick, blue] (u) edge (uadd);
    \draw[ultra thick, blue] (u) edge (startU);
    \draw[ultra thick, blue] (startU) edge (startUadd);
    \draw[black] (startU) edge (v2);
    \draw[ultra thick, blue] (v2) edge (v1);
    \draw[ultra thick, blue] (v2) edge (v2add);
    \draw[ultra thick, blue] (v1) edge (v1add);
    \draw[black] (v1) edge (startV);
    \draw[ultra thick, blue] (startV) edge (startVadd1);
    \draw[black] (startV) edge (startVadd2);
    \draw[ultra thick, blue] (startV) edge (v);
    \draw[ultra thick, blue] (v) edge (vadd);
    \draw[black, dotted] (v) edge (u);

\end{tikzpicture}} & $\rightarrow$ & \resizebox{\looseEndSize}{!}{\begin{tikzpicture}[scale=0.9]
    \tikzstyle{node}=[draw,circle,inner sep=1pt]
    \tikzstyle{active}=[draw,circle,inner sep=1pt,ultra thick,red]
    \tikzstyle{legend}=[minimum width=2.5cm,minimum height=1cm]

    \node[node] (v) at (0,0) {};
    \node[right= of v, xshift=-3em] {\tiny{$v$}};
    \node[node] (vadd) at (0,-0.5) {};

    \node[node] (startV) at (-1,0) {};
    \node[node] (startVadd1) at (-1,-0.5) {};
    \node[node] (startVadd2) at (-1.3,-0.5) {};

    \node[node] (v1) at (-2,1) {};
    \node[node] (v1add) at (-2.5,0.8) {};

    \node[node] (v2) at (-2,2) {};
    \node[node] (v2add) at (-2.5,2.2) {};

    \node[node] (startU) at (-1,3) {};
    \node[node] (startUadd) at (-1,3.5) {};
    
    \node[node] (u) at (0,3) {};
    \node[right= of u, xshift=-3em] {\tiny{$u$}};
    \node[node] (uadd) at (0, 3.5) {};

    \draw[ultra thick, blue] (u) edge (uadd);
    \draw[black] (u) edge node[font=\tiny, above, black] {R} (startU);
    \draw[ultra thick, blue] (startU) edge (startUadd);
    \draw[black] (startU) edge (v2);
    \draw[ultra thick, blue] (v2) edge (v1);
    \draw[ultra thick, blue] (v2) edge (v2add);
    \draw[ultra thick, blue] (v1) edge (v1add);
    \draw[black] (v1) edge (startV);
    \draw[ultra thick, blue] (startV) edge (startVadd1);
    \draw[black] (startV) edge (startVadd2);
    \draw[black] (startV) edge node[font=\tiny, below, black] {R} (v);
    \draw[ultra thick, blue] (v) edge (vadd);
    \draw[ultra thick, blue] (v) edge node[font=\tiny, right, black] {I} (u);

\end{tikzpicture}} \\
            & \textbf{(\romannum{1})} Insert edge $\{u,v\}$ & & \textbf{(\romannum{2})} Maintain \textit{S-invariant} \\
            & & & &\\
            $\rightarrow$ & \resizebox{\looseEndSize}{!}{\begin{tikzpicture}[scale=0.9]
    \tikzstyle{node}=[draw,circle,inner sep=1pt]
    \tikzstyle{active}=[draw,circle,inner sep=1pt,ultra thick,red]
    \tikzstyle{legend}=[minimum width=2.5cm,minimum height=1cm]

    \node[node] (v) at (0,0) {};
    \node[right= of v, xshift=-3em] {\tiny{$v$}};
    \node[node] (vadd) at (0,-0.5) {};

    \node[node] (startV) at (-1,0) {};
    \node[above right= of startV, yshift=-2.9em, xshift=-3.5em] {\tiny{$p^{(u)}_3$}};
    \node[node] (startVadd1) at (-1,-0.5) {};
    \node[node] (startVadd2) at (-1.3,-0.5) {};

    \node[node] (v1) at (-2,1) {};
    \node[node] (v1add) at (-2.5,0.8) {};

    \node[node] (v2) at (-2,2) {};
    \node[node] (v2add) at (-2.5,2.2) {};

    \node[node] (startU) at (-1,3) {};
    \node[below right= of startU, yshift=2.9em, xshift=-3.5em] {\tiny{$p^{(u)}_0$}};
    \node[node] (startUadd) at (-1,3.5) {};
    
    \node[node] (u) at (0,3) {};
    \node[right= of u, xshift=-3em] {\tiny{$u$}};
    \node[node] (uadd) at (0, 3.5) {};

    \draw[ultra thick, blue] (u) edge (uadd);
    \draw[black] (u) edge (startU);
    \draw[ultra thick, blue] (startU) edge (startUadd);
    \draw[ultra thick, blue] (startU) edge node[font=\tiny, above left, black] {I} (v2);
    \draw[black] (v2) edge node[font=\tiny, left, black] {R} (v1);
    \draw[ultra thick, blue] (v2) edge (v2add);
    \draw[ultra thick, blue] (v1) edge (v1add);
    \draw[ultra thick, blue] (v1) edge node[font=\tiny, below left, black] {I} (startV);
    \draw[ultra thick, blue] (startV) edge (startVadd1);
    \draw[black] (startV) edge (startVadd2);
    \draw[black] (startV) edge (v);
    \draw[ultra thick, blue] (v) edge (vadd);
    \draw[ultra thick, blue] (v) edge (u);

    \draw[->] (-1,2.35) arc (120:240:1);

\end{tikzpicture}} & $\rightarrow$ & \resizebox{\looseEndSize}{!}{\begin{tikzpicture}[scale=0.9]
    \tikzstyle{node}=[draw,circle,inner sep=1pt]
    \tikzstyle{active}=[draw,circle,inner sep=1pt,ultra thick,red]
    \tikzstyle{legend}=[minimum width=2.5cm,minimum height=1cm]

    \node[node] (v) at (0,0) {};
    \node[right= of v, xshift=-3em] {\tiny{$v$}};
    \node[node] (vadd) at (0,-0.5) {};

    \node[node] (startV) at (-1,0) {};
    \node[above right= of startV, yshift=-2.9em, xshift=-3.1em] {\tiny{$p^{(v)}_0$}};
    \node[node] (startVadd1) at (-1,-0.5) {};
    \node[node] (startVadd2) at (-1.3,-0.5) {};

    \node[node] (v1) at (-2,1) {};
    \node[right= of v1, xshift=-2.6em, yshift=0.6em] {\tiny{$v_k$}};
    \node[node] (v1add) at (-2.5,0.8) {};

    \node[node] (v2) at (-2,2) {};
    \node[node] (v2add) at (-2.5,2.2) {};

    \node[node] (startU) at (-1,3) {};
    \node[node] (startUadd) at (-1,3.5) {};
    
    \node[node] (u) at (0,3) {};
    \node[right= of u, xshift=-3em] {\tiny{$u$}};
    \node[node] (uadd) at (0, 3.5) {};

    \draw[ultra thick, blue] (u) edge (uadd);
    \draw[black] (u) edge (startU);
    \draw[ultra thick, blue] (startU) edge (startUadd);
    \draw[ultra thick, blue] (startU) edge (v2);
    \draw[black] (v2) edge (v1);
    \draw[ultra thick, blue] (v2) edge (v2add);
    \draw[ultra thick, blue] (v1) edge (v1add);
    \draw[black] (v1) edge node[font=\tiny, below left, black] {R} (startV);
    \draw[ultra thick, blue] (startV) edge (startVadd1);
    \node[right= of startVadd1, xshift=-5.5em, yshift=0.0em] {\tiny{$v_l$}};
    \draw[ultra thick, blue] (startV) edge node[font=\tiny, left, black] {I} (startVadd2);
    \draw[black] (startV) edge (v);
    \draw[ultra thick, blue] (v) edge (vadd);
    \draw[ultra thick, blue] (v) edge (u);

    \draw[->, dotted] (-1,0.2) arc (0:87:0.83);
    \draw[->] (-1.45,0.1) arc (80:210:0.33);

    \draw[black, dotted] (-2.15,0.85) rectangle ++(0.3,0.3);

\end{tikzpicture}}\\
            & \textbf{(\romannum{3})} Traverse path $P_u$ & & \textbf{(\romannum{4})} Traverse path $P_v$ \\
        \end{tblr}
        \caption{Example of the identification of a loose end for edge insertions ($b=2$). $I$ denotes an insertion into $M^{(i)}$, while $R$ represents a removal from $M^{(i)}$. From left to right: \textbf{(\romannum{1})} Insert edge $\{u,v\}$ with $\omega(u,v) > \max\{\omega(u,S(u).min), \omega(v,S(v).min)\}$, where $u$ and $v$ are {saturated}. \textbf{(\romannum{2})} Update suitor lists $S^{(i)}$ of $u$ and $v$ and remove both nodes from the previous $S(u).min$ and $S(v).min$. \textbf{(\romannum{3})} Traverse update path $P_u =(p^{(u)}_0=startU,\dots,p^{(u)}_3=startV)$, which ends at $startV$. \textbf{(\romannum{4})} Traverse update path $P_v$ with $p^{(v)}_0=startV$. Assuming $\omega(p^{(v)}_0,v_l) > \omega(p^{(v)}_0,v_k)$, then $v_l$ is chosen as $p^{(v)}_1$. $P_v$ continues at $v_l$, leaving $v_k$ as {unsaturated} and identifying a loose end.}
        \label{fig:loose_end}
    \end{center}
\end{figure*}

An example for the existence of an update path $P_a$ can be seen in Figure~\ref{fig:loose_end}. The loose end and therefore $P_a$ are detected at the end of either $P_u$ or $P_v$ (depending on which was traversed first). This can be formalized as follows:
\begin{lemma}
    Let $P_u$ and $P_v$ be two update paths resulting from inserting the edge $e = \{u,v\}$ where $P_u$ ends at $startV$ and is traversed before $P_v$. Then an update path $P_a$ is traversed starting at $p^{(u)}_{|P_u|-2}$  iff  the following properties hold: 
    $(i)$ $\omega(p^{(u)}_{|P_u|-2},startV) \leq \omega(startV,S^{(i)}(startV).min)$, $(ii)$ the length of $P_u$ is odd and $(iii)$ $startV$ is {locally affected} after the construction of $P_u$.
    \label{lemma:edge_insertion_loose_end}
\end{lemma}
\begin{proof}
    "$\leftarrow$": Before both update paths $P_u$ and $P_v$ are traversed, both $startU$ and $startV$ become {unsaturated} during \textsc{EdgeInsertion} due to Corollary~\ref{corollary:edge_insertion_both_saturated}.
    Given that $\omega(p^{(u)}_{|P_u|-2},startV) \leq \omega(startV,S^{(i)}(startV).min)$, we can also conclude that $startV$ was {saturated} before the removal of $v$.
    Let $P_u =(p^{(u)}_0=startU,\dots,p^{(u)}_k=startV)$, where $k$ is the length of $P_u$. 
    The last edge $e_{k-1}=\{p^{(u)}_{k-1},p^{(u)}_k\}$ then represents an insertion into $M^{(i)}$ based on Lemma~\ref{lemma:even_odd} and property (ii). As a result, $startV$ gets a new suitor, namely $p^{(u)}_{k-1}$. After the insertion of $p^{(u)}_{k-1}$ into $S^{(i)}(startV)$, $startV$ is {saturated}.
    For the traversal of $P_v$, \textsc{FindAffected} tries to find a new potential suitor for $startV$. Such a suitor is found, since the vertex is still {locally affected} (property (iii)), \ie a vertex $v_l \neq p^{(u)}_{k-1}$ is found with $\omega(startV,v_l) > \omega(v_l, S^{(i)}(v_l).min)$ and $\omega(startV,v_l) > \omega(v_l, p^{(u)}_{k-1})$.
    At the same time, $p^{(u)}_{k-1}$ is the minimum entry in $S^{(i)}(startV).min$, since $\omega(p^{(u)}_{k-1},startV) \leq \omega(startV,S^{(i)}(startV).min)$ due to property (i). As a result, $p^{(u)}_{k-1}$ is removed again due to the {S-invariant} and $P_v$ is continued at the previous $S^{(i)}(v_l).min$.
    Since $p^{(u)}_{k-1}$ is now {unsaturated}, a new update path $P_a$ (possibly of length $0$) occurs with $p^{(a)}_0=p^{(u)}_{k-1}=p^{(u)}_{|P_u|-2}$.

    "$\rightarrow$": Let $P_u =(p^{(u)}_0=startU,\dots,p^{(u)}_k=startV)$, where $k$ is the length of $P_u$. 
    The first vertex $p^{(a)}_0$ of $P_a$ is then $p^{(u)}_{k-1}$. In order to be a starting point of that update path, $p^{(a)}_0$ has to be a not \texttt{None} vertex in line~\ref{algline:findaffected_loose_end} in \textsc{FindAffected}.
    Such a vertex is the previous minimum suitor of some vertex $cu$ (line~\ref{algline:s-invariant1} in \textsc{FindAffected}) during the traversal of $P_v$. This also implies that $cu$ is {saturated} and {locally affected} after the traversal of $P_u$ (property (iii)).
    As a result, $cu$ has to be the first vertex $startV$ on $P_v$, otherwise it is {unsaturated} (see lines~\ref{algline:prevca_s_invariant} and~\ref{algline:prevca_cu}).

    For $startV$ being {saturated} at the beginning of the construction of $P_v$, the previous operation concerning the suitor lists of $startV$ has to be an insertion, \ie an edge $e_{+}$ with an even index on another update path.
    Since $P_u$ has been traversed before $P_v$, $e_{+}$ has to belong to $P_u$ with its endpoints being $\{p^{(a)}_0, startV\}$.
    $e_{+}$ is therefore also the last edge on $P_u$ and $|P_u|$ is odd (property (ii)).
    The weight of $e_{+}$ has to be smaller than or equal to $\omega(startV,S^{(i)}(startV).min)$, otherwise $p^{(a)}_0$ would not have been removed from $S^{(i)}(startV)$ (property (i)).
    \qed
\end{proof}

Note that the additional update path $P_a$ also occurs if $P_v$ is traversed before $P_u$. This leads to Corollary~\ref{corollary:pa_reversed}, which we include for completeness. The proof is easily done by swapping $u$ and $v$ in the proof of Lemma~\ref{lemma:edge_insertion_loose_end}.

\begin{corollary}
    Let $P_u$ and $P_v$ be two update paths resulting from inserting the edge $e = \{u,v\}$ where $P_v$ ends at $startU$ and is traversed before $P_u$. Then an update path $P_a$ is traversed starting at $p^{(v)}_{|P_v|-2}$  iff the following properties hold: 
    $(i)$ $\omega(p^{(v)}_{|P_v|-2},startU) \leq \omega(startU,S^{(i)}(startU).min)$, $(ii)$ the length of $P_v$ is odd and $(iii)$ $startU$ is {locally affected} after the construction of $P_v$.
    \label{corollary:pa_reversed}
\end{corollary}

Now that we have shown that (non-empty) $P_a$ may exist, we still need to prove that no other paths can occur. Also all \laff\ vertices need to be found by \edgeinsParam{u}{v}\ in order to let \dynb\ compute the same matching as \staticb.

\begin{lemma}
    Let $e := \{u,v\} \notin E$ and $\omega(u,v) > \max\{\omega(u,S(u).min), $  \\$\omega(v,S(v).min)\}$. If both $u$ and $v$ are {saturated}, then all \laff\ vertices
    are reached during \edgeinsParam{u}{v}\ via edges from $X := P_u \cup P_v \cup P_a \cup \{u,v\} \cup \{u, startU\} \cup \{v, startV\}$. Moreover, the vertices that are not locally affected do not change their suitor lists.
    \label{lemma:edge_insertion_both_saturated}
\end{lemma}
\begin{proof}
    \drop{We first show that \laff\ vertices reached by \edgeinsParam{u}{v}\ \\ are connected by edges in $X := P_u \cup P_v \cup P_a \cup \{u,v\} \cup \{u, startU\} \cup \{v, startV\}$.
    If an edge insertion is occuring and $\omega(u,v) > \max\{\omega(u,S(u).min),$ $\omega(v,S(v).min)\}$, we already know that $u$ and $v$ are \sat.  As a result, $startU$ and $startV$ are not \none\ in \edgeins. They also mark the first vertices of update paths $P_u$ and $P_v$ according in \findaff\ (Corollary~\ref{corollary:edge_insertion_both_saturated}).
    Due to Lemma~\ref{lemma:edge_insertion_loose_end} and Corollary~\ref{corollary:pa_reversed}, the additional update path $P_a$ either starts at a vertex in $P_u$ or in $P_v$.
    Due to Lemma~\ref{TODO}, all vertices in $P_u$, $P_v$, and $P_a$ except the respective last ones are \laff. }

    \drop{Let $P_c$ denote the path $P_c := \{startU, u\} \cup \{u,v\} \cup \{v, startV\}$, which connects $P_u$, $P_v$, and $P_a$.
    We know that $u$ and $v$ are \laff\ due to $\omega(u,v) > \max\{\omega(u,S(u).min),$ $\omega(v,S(v).min)\}$.
    Therefore, all \laff\ vertices reached by \edgeins\ are connected by edges in $P_u \cup P_v \cup P_a \cup P_c = X$.}

    We first show that no update paths other than $P_u$, $P_v$ and $P_a$ can occur in \edgeinsParam{u}{v}.
    We will do this by an indirect proof.
    Assume we have an additional update path, \ie another $looseEnd$ besides the one in Lemma~\ref{lemma:edge_insertion_loose_end} is detected in line~\ref{algline:findaffected_loose_end} of \textsc{FindAffected}.
    For this to happen, the same vertex has to be assigned to $cu$ and then updated at least twice.
    Let us assume these updates correspond to update paths $P_x$ and $P_y$, where $P_x$ is traversed before $P_y$.
    Specifically for a vertex $v_k$, the update has to be such that after the first update in $P_x$ the vertex is \sat\ but still \laff\
    (line \ref{algline:s-invariant1} in \findaff).
    As a result, the first update represents the insertion of an entry into the suitor list of $S^{(i)}(v_k)$; otherwise $v_k$ cannot be \sat\ after the update.
    In order to get a not \texttt{None} vertex $prevCu$ in line~\ref{algline:s-invariant1}, $cu$ has still to be \sat\ during the second update in $P_y$.
    However, $v_k$ and therefore $cu$ can only be \sat\ during the second update if it is the first vertex on $P_y$. 
    Then it would be equal to $looseEnd$ due to Lemma~\ref{lemma:edge_insertion_loose_end}.
    As a result, \edgeinsParam{u}{v}\ does not detect any additional update path other than $P_u$, $P_v$, and $P_a$.

    The next and final step for the proof is to show that $X$ indeed represents all \laff\ vertices after an edge insertion.
    In order for a vertex to be \laff\, either there exist new neighbors fulfilling Def.~\ref{def:affected} or its suitor list $S^{(i)}$ has changed with respect to $S$.
    The former case is covered in \edgeinsParam{u}{v}\ directly by updating $S^{(i)}(u)$ and $S^{(i)}(v)$.
    For the latter case, the update paths traversed by \findaff\ identify all vertices for which $S^{(i)}$ is changed by searching the respective neighborhood.
    We already covered that there exists no additional update path besides $P_u$, $P_v$, and $P_a$.
    Since \findaff\ traverses via the neighborhood of $cu$, a \laff\ vertex cannot occur outside one of those \uppath s,
    so that no vertex outside of $X$ changes its suitor lists.
    As a result, \dynb\ successfully updates all \laff\ vertices (and only those) after an edge insertion.
    \qed

\end{proof}

We have covered all cases concerning the saturation state of the involved endpoints of the new edge $\{u,v\}$. It remains to show
that the final $b$-matching $M^{(f)}$ indeed equals $\widetilde{M}$, the result of \staticb.

\begin{proposition}
    After \edgeins\ is finished, $M^{(f)}$ equals the $b$-matching $\widetilde{M}$ computed by \staticb\ on $\widetilde{G}$. $\widetilde{M}$ is a (1/2)-approximation of the MWBM problem.
    \label{proposition:edge_insertion_final}  
\end{proposition}
\begin{proof}
    Let $e$ be the edge inserted into $G = (V, E)$ \ie $\widetilde{G}= (V, E \cup \{e\})$. After \staticb\ is finished, Eq.~(\ref{eq:static_suitor}) returns \texttt{None} for every node in $\widetilde{G}$. In Lemmas~\ref{lemma:edge_insertion_unsaturated},~\ref{lemma:edge_insertion_one_saturated}, and~\ref{lemma:edge_insertion_both_saturated} we showed that the update paths for all possible cases in \textsc{EdgeInsertion}$(e)$ together with \textsc{FindAffected} identify and update all {locally affected} nodes successfully.
    
    Hence, after finishing \textsc{EdgeInsertion}$(e)$, for no vertex in $\widetilde{G}$ a better suitor according to Eq.~(\ref{eq:dynamic_suitor}) exists. Due to the tie-breaking rule, the resulting matching $M^{(f)}$ is unique and therefore equals $\widetilde{M}$. Due to Proposition~\ref{proposition:valid_matching_static}, this $b$-matching is a (1/2)-approximation of a  maximum weight $b$-matching.
    \qed
\end{proof}

\subsection{Edge Removal}
\label{sub:proofs-edge-removal}

Using our results from Section~\ref{sub:appendix-edge-insertion}, we now discuss the behavior of \dynb\ 
when the edge update is a removal, \ie $\widetilde{E} = E \setminus \{u,v\}$ with $ \{u,v\} \in E$. 

\begin{lemma}
    Let $\widetilde{G}=(V,\widetilde{E})=(V,E \setminus \{e\})$ with $e=\{u,v\} \in E$. If $e \notin M$, then $u$ and $v$ are not {locally affected}.
    Otherwise $u$ and $v$ are potentially \laff.
    \label{lemma:edge_removal_basic}
\end{lemma}

\begin{proof}
    If $e =\{u,v\}\notin M$, then the removal clearly does not change the suitor lists of $u$/$v$ nor the saturation state.
    Also $\widetilde{N}(u)$ $\left[\widetilde{N}(v)\right]$ does not contain any new neighbor.
    Therefore ${s}$ according to Definition \ref{def:affected} is \texttt{None}; otherwise it would have been
    picked by \staticb\ for $S(u)$ $\left[S(v)\right]$ in the initial $b$-matching.

    Let now $e =\{u,v\}\in M$; since the proof is symmetric, we focus only on $u$.
    After removing $u$ from $S^{(i)}(v)$, $v$ becomes {unsaturated}.
    Before the removal, $u$ might have been either \usat\ or \sat.
    Let us assume, $u$ is \sat. Then, the minimum of $u$ has changed from $S(v).min$ to \texttt{None}.
    As a result, there might exist a vertex $x \in \widetilde{N}(u) \setminus S^{(i)}(u)$ such that $ \omega\left(u,x\right) > \omega\left(x, S^{(i)}(x).min\right)$; then $u$ is affected according to Definition~\ref{def:affected}.
\qed 
\end{proof}

\begin{algorithm}[tb]
    \begin{algorithmic}[1]
        \begin{scriptsize}
            \Procedure{EdgeRemoval}{$\widetilde{G}$, $\{u, v\}$}
            \State \textbf{Input:} Graph $\widetilde{G}=(V,\widetilde{E})$, an edge  $\{u, v\} \in E$
            \State \textbf{Output:} New $(1/2)$-approximate b-Matching $M^{(f)}$
            \If{ $u \in S^{(i)}(v) \textbf{ and } v \in S^{(i)}(u)$} \Comment{Lemma~\ref{lemma:edge_removal_basic}} \label{algline:edge_removal_exclusion}
            \State $S^{(i)}(u).remove(v)$ 
            \State $S^{(i)}(v).remove(u)$ 
            \State $\textsc{FindAffected}(u)$ \label{algline:edge_removal_u}
            \State $\textsc{FindAffected}(v)$ \label{algline:edge_removal_v}
            \State $M^{(i)} \gets \{\{u^\prime,v^\prime \} \in \widetilde{E} \text{ s.t } v^\prime \in S^{(i)}(u^\prime) \wedge u^\prime \in S^{(i)}(v^\prime)\}$
            \EndIf
            \State $M^{(f)} \gets M^{(i)}$
            \EndProcedure
        \end{scriptsize}
    \end{algorithmic}
    \caption{Dynamic Edge Removal}
    \label{alg:dynamic_edge_removal_new}
\end{algorithm}

Similarly to the edge insertion case, we use Lemma~\ref{lemma:edge_removal_basic} to derive the general program flow of \edgerem.
The input is given by an already updated graph with edge set $\widetilde{E} = E \setminus \{u, v\}$ and the removed edge $\{u, v\}$.
Before processing the endpoint nodes of the edge, the algorithm only processes nodes in $M^{(i)}$ (line \ref{algline:edge_removal_exclusion}) exploiting Lemma~\ref{lemma:edge_removal_basic}.
Different from \textsc{EdgeInsertion}, $u$ and $v$ are only \emph{potentially} \laff~-- even if the removed $\{u, v\}$ was originally part of the $b$-matching $M$.
This leads again to three cases -- which address the \laff\ state of the involved endpoints of the removed edge $\{u,v\}$: \textbf{(i)} both $u$ and $v$ are not \laff\ (Lemma \ref{lemma:edge_removal_unaffected}), \textbf{(ii)} only one of them is \laff\ (Lemma \ref{lemma:edge_removal_one_affected}) and \textbf{(iii)} both are \laff\ (Lemma~\ref{lemma:edge_removal_loose_end} and~\ref{lemma:edge_removal_both_affected}).

\subsubsection{Case \textbf{(i)}: Both $u$ and $v$ are not \laff.} 
In this case, no new suitors can be found for $u$ and $v$:
\begin{lemma}
    Let $e = \{u,v\} \in E$ and $\widetilde{G}=(V,\widetilde{E} = E \setminus \{u,v\})$. If both $u$ and $v$ are not \laff\ after the removal of $e$, then $M^{(f)} = \widetilde{M} = M \setminus \{u,v\}$ and no other vertices are locally affected.
    \label{lemma:edge_removal_unaffected}
\end{lemma}

\begin{proof}
    In case both $u$ are $v$ are not {locally affected}, \ie ${s}(u) = {s}(v) = \texttt{None}$ according to Definition~\ref{def:affected}, the matching produced by \staticb\ is $\widetilde{M} = M \setminus \{u,v\}$.
    Moreover, this also means that both calls to \findaff\ (in lines~\ref{algline:edge_removal_u} and~\ref{algline:edge_removal_v} of \edgeremParam{u}{v}) do not find potential new partners for $u$ and $v$  (line~\ref{algline:check_def_affected} in \textsc{FindAffected}).
    As a result, only $S^{(i)}(u)$ and $S^{(i)}(v)$ are changed and therefore $M^{(f)} = \widetilde{M}$.
\qed
\end{proof}

\subsubsection{Case \textbf{(ii)}: One endpoint is \laff} 
Similarly to \edgeins, after removing $u$ from $S^{(i)}(u)$ and vice versa, \edgerem\ calls \findaff\ for both endpoints.
If only one vertex is \laff, then one of the traversed update paths is empty.
The other update path is similar to the one in Lemmma~\ref{lemma:edge_insertion_one_saturated};
the only change is that the \uppath\ starts at the \laff\ node $u$ [$v$] instead of $startU$ [$startV$].

\begin{lemma}
    Let $e = \{u,v\} \in E$ and $\widetilde{G}=(V,\widetilde{E} = E \setminus \{u,v\})$. If only $u$ $[v]$ is \laff\ after the edge removal,
    \edgeremParam{u}{v}\ traverses one non-empty update path $P_u$ $\left[P_v\right]$.
    \label{lemma:edge_removal_one_affected}
\end{lemma}

\begin{proof}
    Since the proof is symmetric for both cases, we assume w.l.o.g.\ that $u$ is the \laff\ vertex, while $v$ is not \laff.
    According to Definition~\ref{def:affected} and~\ref{definition:update_path}, $P_v$ is an empty \uppath, since it only consists of its starting vertex $v$ (as called from \edgerem). 
    Since $u$ is \laff, $P_u$ consists at least of two vertices: starting vertex $u$ (as called from \edgeremParam{u}{v}) and a not \texttt{None} vertex $ca$ as identified in line~\ref{algline:check_def_affected}. In case $ca$ is already \sat\ before maintaining the {S-invariant} in \findaff, the path continues in the same way as described in Lemma~\ref{lemma:edge_insertion_one_saturated}. \drop{As a result, the edge weights are decreasing in $P_u$.}
\qed
\end{proof}

\subsubsection{Case \textbf{(iii)}: Both $u$ and $v$ are \laff} 

Naturally, if $u$ and $v$ are both \laff\ after the edge removal, we traverse at least two non-empty update paths $P_u$ and $P_v$. 
Analogously to edge insertion, we can derive this as a consequence of Lemma~\ref{lemma:edge_removal_one_affected}.

\begin{corollary}
    Let $e = \{u,v\} \in E$ and $\widetilde{G}=(V,\widetilde{E} = E \setminus \{u,v\})$. If both $u$ and $v$ are \laff\ after the edge removal, then two non-empty update paths $P_u$ and $P_v$ are traversed; they start at $u$ and $v$, respectively.
    \label{corollary:edge_removal_both_affected}
\end{corollary}

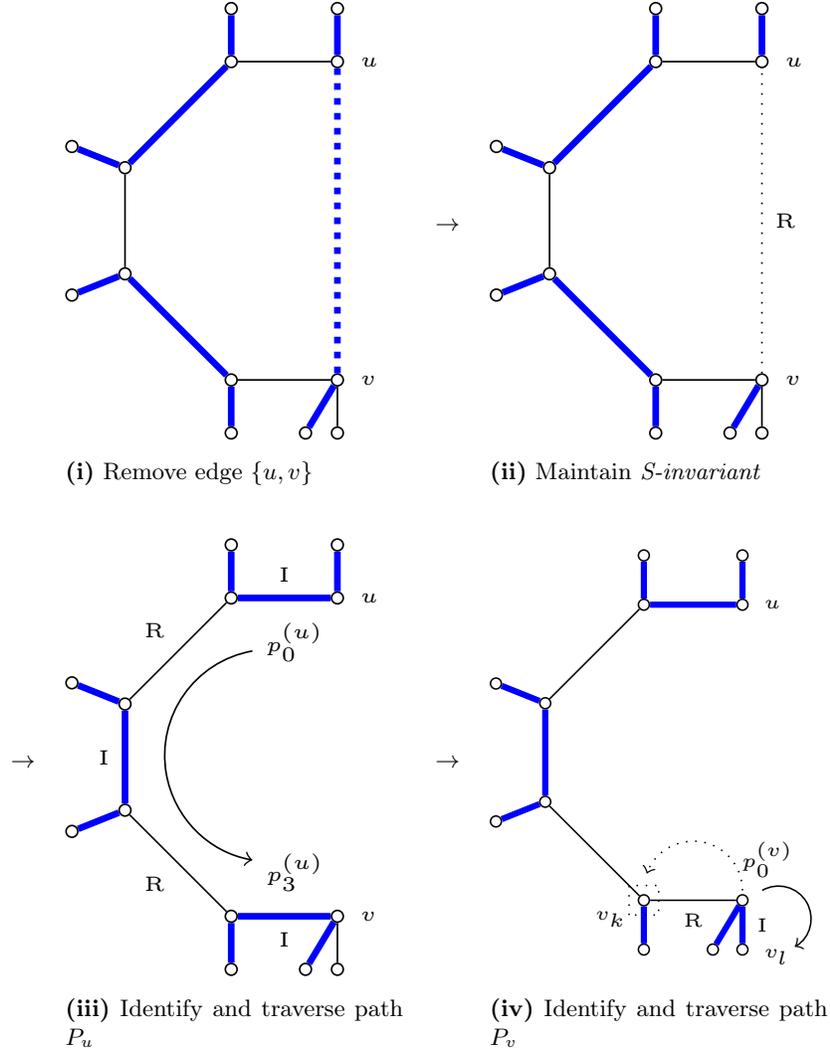
\begin{figure*}
    \begin{center}
        \begin{tblr}{Q[0.3cm,valign=m] Q[\looseEndSize, valign=h] Q[0.3cm,valign=m] Q[\looseEndSize, valign=h]}
            & \resizebox{\looseEndSize}{!}{\begin{tikzpicture}[scale=0.9]
    \tikzstyle{node}=[draw,circle,inner sep=1pt]
    \tikzstyle{active}=[draw,circle,inner sep=1pt,ultra thick,red]
    \tikzstyle{legend}=[minimum width=2.5cm,minimum height=1cm]

    \node[node] (v) at (0,0) {};
    \node[right= of v, xshift=-3em] {\tiny{$v$}};
    \node[node] (vadd) at (0,-0.5) {};
    \node[node] (vadd2) at (-0.3,-0.5) {};

    \node[node] (startV) at (-1,0) {};
    \node[node] (startVadd1) at (-1,-0.5) {};

    \node[node] (v1) at (-2,1) {};
    \node[node] (v1add) at (-2.5,0.8) {};

    \node[node] (v2) at (-2,2) {};
    \node[node] (v2add) at (-2.5,2.2) {};

    \node[node] (startU) at (-1,3) {};
    \node[node] (startUadd) at (-1,3.5) {};
    
    \node[node] (u) at (0,3) {};
    \node[right= of u, xshift=-3em] {\tiny{$u$}};
    \node[node] (uadd) at (0, 3.5) {};

    \draw[ultra thick, blue] (u) edge (uadd);
    \draw[black] (u) edge (startU);
    \draw[ultra thick, blue] (startU) edge (startUadd);
    \draw[ultra thick, blue] (startU) edge (v2);
    \draw[black] (v2) edge (v1);
    \draw[ultra thick, blue] (v2) edge (v2add);
    \draw[ultra thick, blue] (v1) edge (v1add);
    \draw[ultra thick, blue] (v1) edge (startV);
    \draw[ultra thick, blue] (startV) edge (startVadd1);
    \draw[black] (startV) edge (v);
    \draw[black] (v) edge (vadd);
    \draw[ultra thick, blue] (v) edge (vadd2);
    \draw[ultra thick, blue, dotted] (v) edge (u);

\end{tikzpicture}} & $\rightarrow$ & \resizebox{\looseEndSize}{!}{\begin{tikzpicture}[scale=0.9]
    \tikzstyle{node}=[draw,circle,inner sep=1pt]
    \tikzstyle{active}=[draw,circle,inner sep=1pt,ultra thick,red]
    \tikzstyle{legend}=[minimum width=2.5cm,minimum height=1cm]

    \node[node] (v) at (0,0) {};
    \node[right= of v, xshift=-3em] {\tiny{$v$}};
    \node[node] (vadd) at (0,-0.5) {};
    \node[node] (vadd2) at (-0.3,-0.5) {};

    \node[node] (startV) at (-1,0) {};
    \node[node] (startVadd1) at (-1,-0.5) {};

    \node[node] (v1) at (-2,1) {};
    \node[node] (v1add) at (-2.5,0.8) {};

    \node[node] (v2) at (-2,2) {};
    \node[node] (v2add) at (-2.5,2.2) {};

    \node[node] (startU) at (-1,3) {};
    \node[node] (startUadd) at (-1,3.5) {};
    
    \node[node] (u) at (0,3) {};
    \node[right= of u, xshift=-3em] {\tiny{$u$}};
    \node[node] (uadd) at (0, 3.5) {};

    \draw[ultra thick, blue] (u) edge (uadd);
    \draw[black] (u) edge (startU);
    \draw[ultra thick, blue] (startU) edge (startUadd);
    \draw[ultra thick, blue] (startU) edge (v2);
    \draw[black] (v2) edge (v1);
    \draw[ultra thick, blue] (v2) edge (v2add);
    \draw[ultra thick, blue] (v1) edge (v1add);
    \draw[ultra thick, blue] (v1) edge (startV);
    \draw[ultra thick, blue] (startV) edge (startVadd1);
    \draw[black] (startV) edge (v);
    \draw[black] (v) edge (vadd);
    \draw[ultra thick, blue] (v) edge (vadd2);
    \draw[black, dotted] (v) edge node[font=\tiny, right, black] {R} (u);

\end{tikzpicture}} \\
            & \textbf{(\romannum{1})} Remove edge $\{u,v\}$ & & \textbf{(\romannum{2})} Maintain \textit{S-invariant} \\
            & & & &\\
            $\rightarrow$ & \resizebox{\looseEndSize}{!}{\begin{tikzpicture}[scale=0.9]
    \tikzstyle{node}=[draw,circle,inner sep=1pt]
    \tikzstyle{active}=[draw,circle,inner sep=1pt,ultra thick,red]
    \tikzstyle{legend}=[minimum width=2.5cm,minimum height=1cm]

    \node[node] (v) at (0,0) {};
    \node[right= of v, xshift=-3em] {\tiny{$v$}};
    \node[node] (vadd) at (0,-0.5) {};
    \node[node] (vadd2) at (-0.3,-0.5) {};

    \node[node] (startV) at (-1,0) {};
    \node[above right= of v, yshift=-2.9em, xshift=-5.4em] {\tiny{$p^{(u)}_3$}};
    \node[node] (startVadd1) at (-1,-0.5) {};

    \node[node] (v1) at (-2,1) {};
    \node[node] (v1add) at (-2.5,0.8) {};

    \node[node] (v2) at (-2,2) {};
    \node[node] (v2add) at (-2.5,2.2) {};

    \node[node] (startU) at (-1,3) {};
    \node[below left= of u, yshift=2.9em, xshift=3.2em] {\tiny{$p^{(u)}_0$}};
    \node[node] (startUadd) at (-1,3.5) {};
    
    \node[node] (u) at (0,3) {};
    \node[right= of u, xshift=-3em] {\tiny{$u$}};
    \node[node] (uadd) at (0, 3.5) {};

    \draw[ultra thick, blue] (u) edge (uadd);
    \draw[ultra thick, blue] (u) edge node[font=\tiny, above, black] {I} (startU);
    \draw[ultra thick, blue] (startU) edge (startUadd);
    \draw[black] (startU) edge node[font=\tiny, above left, black] {R} (v2);
    \draw[ultra thick, blue] (v2) edge node[font=\tiny, left, black] {I} (v1);
    \draw[ultra thick, blue] (v2) edge (v2add);
    \draw[ultra thick, blue] (v1) edge (v1add);
    \draw[black] (v1) edge node[font=\tiny, below left, black] {R} (startV);
    \draw[ultra thick, blue] (startV) edge (startVadd1);
    \draw[ultra thick, blue] (startV) edge node[font=\tiny, below, black] {I} (v);
    \draw[black] (v) edge (vadd);
    \draw[ultra thick, blue] (v) edge (vadd2);

    \draw[->] (-0.8,2.50) arc (100:260:1);

\end{tikzpicture}} & $\rightarrow$ & \resizebox{\looseEndSize}{!}{\begin{tikzpicture}[scale=0.9]
    \tikzstyle{node}=[draw,circle,inner sep=1pt]
    \tikzstyle{active}=[draw,circle,inner sep=1pt,ultra thick,red]
    \tikzstyle{legend}=[minimum width=2.5cm,minimum height=1cm]

    \node[node] (v) at (0,0) {};
    \node[node] (vadd) at (0,-0.5) {};
    \node[node] (vadd2) at (-0.3,-0.5) {};

    \node[node] (startV) at (-1,0) {};
    \node[above right= of startV, yshift=-2.9em, xshift=-0.8em] {\tiny{$p^{(v)}_0$}};
    \node[node] (startVadd1) at (-1,-0.5) {};

    \node[node] (v1) at (-2,1) {};
    \node[right= of v1, xshift=-2.2em, yshift=-3.3em] {\tiny{$v_k$}};
    \node[node] (v1add) at (-2.5,0.8) {};

    \node[node] (v2) at (-2,2) {};
    \node[node] (v2add) at (-2.5,2.2) {};

    \node[node] (startU) at (-1,3) {};
    \node[node] (startUadd) at (-1,3.5) {};
    
    \node[node] (u) at (0,3) {};
    \node[right= of u, xshift=-3em] {\tiny{$u$}};
    \node[node] (uadd) at (0, 3.5) {};

    \draw[ultra thick, blue] (u) edge (uadd);
    \draw[ultra thick, blue] (u) edge (startU);
    \draw[ultra thick, blue] (startU) edge (startUadd);
    \draw[black] (startU) edge (v2);
    \draw[ultra thick, blue] (v2) edge (v1);
    \draw[ultra thick, blue] (v2) edge (v2add);
    \draw[ultra thick, blue] (v1) edge (v1add);
    \draw[black] (v1) edge (startV);
    \draw[ultra thick, blue] (startV) edge (startVadd1);
    \node[right= of vadd, xshift=-3.0em, yshift=-0.2em] {\tiny{$v_l$}};
    \draw[ultra thick, blue] (v) edge node[font=\tiny, right, black] {I} (vadd);
    \draw[black] (startV) edge node[font=\tiny, below, black] {R} (v);
    \draw[ultra thick, blue] (v) edge (vadd2);

    \draw[->, dotted] (0,0.1) arc (0:160:0.5);
    \draw[->] (0.2,0.1) arc (120:-60:0.33);

    \draw[black, dotted] (-1.15,-0.15) rectangle ++(0.3,0.3);

\end{tikzpicture}}\\
            & \textbf{(\romannum{3})} Identify and traverse path $P_u$ & & \textbf{(\romannum{4})} Identify and traverse path $P_v$ \\
        \end{tblr}
        \caption{Example of the creation of a loose end for edge removal ($b=2$). $I$ denotes an insertion into $M^{(i)}$, while $R$ represents a removal from $M^{(i)}$. From left to right: \textbf{(\romannum{1})} Remove edge $\{u,v\}$.
        \textbf{(\romannum{2})} Update suitor lists $S^{(i)}$ of $u$ and $v$ and remove both nodes from the previous $S(u).min$ and $S(v).min$. $u$ and $v$ are \laff\ after the removal. \textbf{(\romannum{3})} Identify and traverse update path $P_u =(p^{(u)}_0=u,\dots,p^{(u)}_5=v)$, which ends at $v$. \textbf{(\romannum{4})} Identify and traverse update path $P_v$ with $p^{(v)}_0=v$. Assume $\omega(p^{(v)}_0,v_l) > \omega(p^{(v)}_0,v_k)$, then $v_l$ is chosen as $p^{(v)}_1$. $P_v$ continues at $v_l$, leaving $v_k$ \usat\ and creating a loose end.}
        \label{fig:removal_loose_end}
    \end{center}
\end{figure*}

Following the same arguments as for edge insertion, a potential overlap between $P_u$ and $P_v$ can lead to a situation where an update path leads to another update path. 
This is again detected by \findaff\ in line~\ref{algline:findaffected_loose_end}, where the starting point of this additional update path is initalized as $looseEnd$. Note that this situation is again very similar to Lemma~\ref{lemma:edge_insertion_loose_end} on edge insertion. As a result, we again refer to this additional update path as $P_a$. An example of an update path $P_a$ during edge removal can be seen in Figure~\ref{fig:removal_loose_end}.

\begin{lemma}
    Let $P_u$ and $P_v$ be two update paths resulting from the removal of $e = \{u,v\}$ where $P_u$ ends at $v$ and is traversed before $P_v$. Then an update path $P_a$ is traversed starting at $p^{(u)}_{|P_u|-2}$  iff  the following properties hold: 
    $(i)$ $\omega(p^{(u)}_{|P_u|-2},v) \leq \omega(v,S^{(i)}(v).min)$, $(ii)$ the length of $P_u$ is odd and $(iii)$ $v$ is \sat\ after traversing $P_u$.
    \label{lemma:edge_removal_loose_end}
\end{lemma}
\begin{proof}
The proof for the existence is nearly identical to the proof for Lemma~\ref{lemma:edge_insertion_loose_end}.
One just needs to exchange $startU$ with $u$ and $startV$ with $v$.  
\qed 
\end{proof}

\begin{lemma}
    Let $e = \{u,v\} \in E$ and $\widetilde{G}=(V,\widetilde{E} = E \setminus \{u,v\})$. If both $u$ and $v$ are \laff\ after the edge removal, then all \laff\ vertices are reached during \edgeremParam{u}{v}\ via edges from $X := P_u \cup P_v \cup P_a$. Moreover, the vertices that are not locally affected do not change their suitor lists.
    \label{lemma:edge_removal_both_affected}
\end{lemma}

\begin{proof}
    We use the same proof structure as for Lemma~\ref{lemma:edge_insertion_both_saturated}, \ie 
    (i) we have to prove that no other update paths besides $P_u$, $P_v$ and $P_a$ exist, and (ii) that $X$ indeed represents all \laff\ vertices after an edge removal.
    
    The proof for (i) can be directly transferred from Lemma~\ref{lemma:edge_insertion_both_saturated}.
    We already know that $P_u$, $P_v$, and $P_a$ are the only update paths traversed when \findaff\ is called twice for two different vertices. As a result, \edgeremParam{$u$}{$v$} only detects $P_u$, $P_v$, and $P_a$.

    Concerning (ii): In order for a vertex to be \laff, there either exist new neighbors fulfilling Definition~\ref{def:affected} or its suitor list $S^{(i)}$ has changed with respect to $S$.
    The former case is covered in \edgeremParam{$u$}{$v$} by calling \findaff\ for both $u$ and $v$. Since both vertices are known to be \laff, there exists a new suitor in their neighborhood, leading to updates to both $S^{(i)}(u)$ and $S^{(i)}(v)$.
    For the latter case, the update paths traversed by \findaff\ consequently identify all vertices for which $S^{(i)}$ is changed by searching the neighborhood.
    We already covered in the proof of Lemma~\ref{lemma:edge_insertion_both_saturated} that there exists no additional update path besides $P_u$, $P_v$, and $P_a$.
    Since \findaff\ traverses via the neighborhood of $cu$, a \laff\ vertex cannot occur outside one of those \uppath s,
    so that no vertex outside of $X$ changes its suitor lists.
    As a result, \dynb\ successfully updates all \laff\ vertices (and no others) after an edge removal.
    \qed
\end{proof}

\begin{proposition}
    After \edgerem\ is finished, $M^{(f)}$ equals the matching $\widetilde{M}$ computed using \staticb\ on $\widetilde{G}$. $\widetilde{M}$ is a (1/2)-approximation of the MWBM problem.
    \label{proposition:edge_removal_final}  
\end{proposition}
\begin{proof}
    Let $e$ be the edge removed from $G = (V, E)$ \ie $\widetilde{G}= (V, E \setminus \{e\})$. After \staticb\ is finished, Eq.~(\ref{eq:static_suitor}) returns \texttt{None} for every node in $\widetilde{G}$. In Lemmas~\ref{lemma:edge_removal_unaffected},~\ref{lemma:edge_removal_one_affected}, and~\ref{lemma:edge_removal_both_affected} we showed that the update paths for all possible cases in \textsc{EdgeRemoval}$(e)$ together with \textsc{FindAffected} identify and update all \laff\ vertices successfully.
    
    Hence, after finishing \textsc{EdgeRemoval}$(e)$, for no vertex in $\widetilde{G}$ a better suitor according to Eq.~(\ref{eq:dynamic_suitor}) exists. Due to the tie-breaking rule, the resulting matching $M^{(f)}$ is unique and therefore equals $\widetilde{M}$. Following from Proposition~\ref{proposition:valid_matching_static}, this $b$-matching is a (1/2)-approximation of a  maximum weight $b$-matching.
    \qed
\end{proof}

\subsection{\dynb\ -- Batch Updates}
\label{sub:appendix_batch_updates}
As already mentioned in Section~\ref{sub:dynamic_suitor}, \dynb\ also supports batch edge updates by applying multiple edge updates $B=\{e_1, \dots, e_i\}$ in a loop using \edgeins\ (or \edgerem, respectively) on $\widetilde{G}$.
To this end, we define two additional functions: \bedgeins\ (Algorithm~\ref{alg:dynamic_suitor_batch_insertion}) and \bedgerem\ (Algorithm~\ref{alg:dynamic_suitor_batch_removal}).

\begin{algorithm}[tb]
    \begin{algorithmic}[1]
        \begin{scriptsize}
            \Procedure{batchEdgeInsertion}{$G$, $B$}
            \State \textbf{Input:} Graph $G=(V,E)$, a set of edges $B \notin E$
            \State \textbf{Output:} New $(1/2)$-approximate \drop{cardinality} $b$-matching $M^{(f)}$
            \State $\widetilde{E} \gets E$
            \For {$\{u,v \} \in B$}
            \If{$\omega(u,v) > \max\{\omega(u,S(u).min), \omega(v,S(v).min)\}$} \Comment{Lemma \ref{lemma:new_edge_def_1_new}}
            \State $\widetilde{E} \gets \widetilde{E} \cup \{u,v\}$
            \State $\textsc{EdgeInsertion}(\widetilde{G}=(V,\widetilde{E}), \{u,v\})$
            \EndIf
            \EndFor
            \State $M^{(f)} \gets \{\{u^\prime,v^\prime \} \in \widetilde{E} \text{ s.t } v^\prime \in S^{(f)}(u^\prime) \wedge u^\prime \in S^{(f)}(v^\prime)\}$
            \EndProcedure
        \end{scriptsize}
    \end{algorithmic}
    \caption{Dynamic Batch Edge Insertion}
    \label{alg:dynamic_suitor_batch_insertion}
\end{algorithm}

The process of handling multiple edge updates builds upon previous work by Angriman et al.~\cite{angriman2022batch} for
$1$-matchings. Both algorithms benefit from similar time-saving properties.
While the time complexity is determined by the batch size, the empirical running time may benefit from synergies 
by updates non-empty overlaps between \uppath\ of different edge updates.

\begin{algorithm}[tb]
    \begin{algorithmic}[1]
        \begin{scriptsize}
            \Procedure{batchEdgeRemoval}{$G$, $B$}
            \State \textbf{Input:} Graph $G=(V,E)$, a set of edges $B \in E$
            \State \textbf{Output:} New 1/2-approximate \drop{cardinality} $b$-matching $M^{(f)}$
            \State $\widetilde{E} \gets E$
            \For {$\{u,v \} \in B$}
            \State $\widetilde{E} \gets \widetilde{E} \setminus \{u,v\}$
            \State $\textsc{EdgeRemoval}(\widetilde{G}=(V,\widetilde{E}), \{u,v\})$
            \EndFor
            \State $M^{(f)} \gets \{\{u^\prime,v^\prime \} \in \widetilde{E} \text{ s.t } v^\prime \in S^{(f)}(u^\prime) \wedge u^\prime \in S^{(f)}(v^\prime)\}$
            \EndProcedure
        \end{scriptsize}
    \end{algorithmic}
    \caption{Dynamic Batch Edge Removal}
    \label{alg:dynamic_suitor_batch_removal}
\end{algorithm}

\newpage

\subsection{\dynb\ - Proof of Theorem~\ref{theorem:dynb_final}}
\label{sub:appx-proof-main-thm}
Recall that Theorem~\ref{theorem:dynb_final} states that \dynb\ computes the same $b$-matching as
\staticb. Due to space constraints, we provide the proof here, using the results from
Sections~\ref{sub:appendix-edge-insertion} to~\ref{sub:appendix_batch_updates}.
\begin{proof}
    \dynb\ is designed to support both edge insertions and removals for either single-edge updates or batches of updates. We have shown in Propositions~\ref{proposition:edge_insertion_final} and~\ref{proposition:edge_removal_final} that for the single-edge update case, \dynb\ produces the same $b$-matching as \staticb\ for both edge insertion and edge removal. For batches of updates, \bedgerem\ and \bedgeins\ use the logic of \edgerem\ and \edgeins\ from the single-edge update case in a loop.
    $M^{(f)}$ is set after all edges of the batch have been processed.
    As a result, when \dynb\ is finished, $M^{(f)}$ equals $\widetilde{M}$.
\end{proof}

\newpage

\section{Instance Statistics}
\label{sub:instance-stats}
\vspace{-0.5cm}
\begin{table*}[h!]
    \begin{center}
        \begin{tabular}{|| c l r r r||}
            \hline
            & \textbf{graph}        & \textbf{nodes} & \textbf{edges} & \textbf{avg. degree}\\ [0.5ex]
            \hline\hline
            \hspace{0.1cm}\multirow{5}{0.4cm}{\rotatebox[origin=c]{90}{\textit{Sparse}}} & human\_gene2          & $14340$        & $9027024$      & 1259                 \\
            & mouse\_gene           & $45101$        & $14461095$     & 641.28               \\
            & cage14                & $1505785$      & $25624564$     & 17.02                \\
            & bone10                & $986703$       & $35339811$     & 71.63                \\
            & cage15                & $5154859$      & $9444692$      & 18.24                \\
            \hline
            \hspace{0.1cm}\multirow{6}{0.4cm}{\rotatebox[origin=c]{90}{\textit{Infra}}} & belgium               & $1216902$      & $1563642$      & 2.57                 \\
            & austria               & $2621866$      & $3082590$      & 2.35                 \\
            & poland                & $5567642$      & $7200814$      & 2.59                 \\
            & france                & $11063911$     & $13785539$     & 2.49                 \\
            & africa                & $23975266$     & $31044959$     & 2.59                 \\
            & asia                  & $57736107$     & $72020649$     & 2.49                 \\
            \hline
            \hspace{0.1cm}\multirow{5}{0.4cm}{\rotatebox[origin=c]{90}{\textit{Social}}} & com-youtube           & $1134890$      & $2987624$      & 5.26                 \\
            & petster-catdog-friend & $623754$       & $13991746$     & 15.86                \\
            & flickr-growth         & $2302925$      & $22838276$     & 19.83                \\
            & soc-LiveJournal1      & $4846609$      & $68475391$     & 14.12                \\
            & orkut-links           & $3072441$      & $117184899$    & 76.28                \\
            \hline
            \hspace{0.1cm}\multirow{4}{0.4cm}{\rotatebox[origin=c]{90}{\textit{Random}}} & rmat-er               & $1048576$      & $8388557$      & 16.00                \\
            & rmat-b                & $1048576$      & $8281690$      & 15.80                \\
            & rmat-g                & $1048576$      & $8377324$      & 15.98                \\
            & hyp-24-1              & $16777216$     & $167676501$    & 19.99                \\
            \hline
        \end{tabular}
    \end{center}
    \caption{Datasets used for experimental results}
    \label{tab:datasets}
\end{table*}

\vspace{-0.6cm}

All networks from the \emph{Sparse} category are obtained from the Suite Sparse Matrix Collection~\cite{davis2011university}.
Note that the directed networks \texttt{cage14} and \texttt{cage15} are converted into an undirected representation in which the edge weight corresponds to the weight of the directed edge $(u, v)$ if $u > v$, while the edge $(v, u)$ is dropped.
In addition, we use the same road networks (based on OpenStreetMap) for the \emph{Infra} category that were chosen by Angriman et al.~\cite{angriman2022batch};
the edge weights therein represent geographic distances.
The networks for the \emph{Social} category are obtained from KONECT~\cite{kunegis2013konect} and their edges are assigned with uniformly random generated edge weights between $0.0$ and $1.0$, as only the connectivity of the initial graph is provided.
The networks for the \emph{Random} category are generated using the R-MAT generator from NetworKit~\cite{angriman2023algorithms}.
R-MAT uses four probability parameters to steer the edge generation process; we use the same parameter values as in Refs.~\cite{halappanavar2012approximate,khorasani2015scalable}:
(i) the setting equivalent to the $G(n,p)$ model \texttt{rmat-er} with equal probabilities $(0.25, 0.25, 0.25, 0.25)$\drop{ for an edge to occur},
(ii) \texttt{rmat-g} $(0.45, 0.15, 0.15, 0.25)$, and (iii) \texttt{rmat-b} $(0.55, 0.15, 0.15, 0.15)$.
The generator uses a scale parameter of $20$ and an edge factor of $8$, resulting in graphs with $2^{20}$ vertices and approximately $8 \times 2^{20}$ edges.
In addition we include a network based on the random hyperbolic model, resembling real-world complex networks. For this the generator from von 
Looz \etal\cite{von2016generating} within NetworKit is used. The average degree is set to $20$, and the exponent of the power-law distribution to~$3$.
The edge weights \drop{(in case they are unweighted)} follow a uniform continuous distribution.

\newpage

\section{Additional Experimental Results}
\label{sub:add-exp-results}

\subsection{Locally Affected Vertices}
\label{subsub:exp_results_affected}

The main influence for the algorithmic speedup of the dynamic algorithm is the number of locally affected vertices, reflecting the average update path length.
In the following we examine the influence of both the batch size and value of $b$ on the total number of locally affected nodes for the categories \emph{Sparse}, \emph{Infra}, \emph{Social} and \emph{Random}.

The results are shown in Appendix~\ref{subsub:exp_results_additional} in Figures~\ref{fig:res_aff_insertion_misc} (edge insertion) and~\ref{fig:res_aff_removal_misc} (edge removal), respectively.
For single edge updates (insertion and removal combined), we see some variance in the maximum number of nodes (max path length), with networks from the \emph{Sparse} category being the outlier with $31$ locally affected nodes. The other categories are all below $10$.
This might be misleading, however, since for single updates the path length also depends on the chosen edge and the distribution of the min-entries in the suitor lists.
When randomly drawing an edge from the network, its weight is likely to be near the median.
Since \staticb\ tries to maximize the overall matching weight, the $min$-entries in the suitor lists are more likely to be closer to the average, which is for most networks above the median.
It is therefore likely that a randomly chosen edge does not {necessarily} influence the overall matching.
Also other factors, such as the average degree, play a role here, especially for small values of $b$.
As a result, the average path length for single edge updates is for all categories below $1$ (avg: $0.65$), with the highest value of $1.5$ for networks from the \emph{Infra} category. This is in line with the expectation based on the findings for the dynamic Suitor algorithm and given these networks have the lowest average node degree from the four categories.
With increasing batch size, the number of locally affected vertices increases almost linearly, with an average of $680$ for batch size $10^3$.

\subsection{Comparison With the Parallel \staticb}
\label{sub:exp_results_parallel}

For completeness, we also compare against the parallel implementation of \staticb, which is provided in the open-source framework ExaGraph~\cite{khan2016efficient}.
Note that the parallel implementation shows good strong scaling behavior~\cite{khan2016efficient}.
This is in line with our experimental results. On an AMD Epyc 9754 with 128 cores, the code achieves a parallel speedup of $60\times$ compared to our sequential implementation of \staticb.
As shown in Figure~\ref{fig:res_speed_par_misc}, when comparing to \dynb\, for single-edge updates, the average (geometric mean) acceleration by \dynb (= algorithmic speedup) is $3.7\cdot 10^4$ over all instances. Networks from the \emph{Infra} and \emph{Random} category show the lowest speedup with $2.3\cdot 10^4$ and $1.4\cdot 10^4$, respectively.

For batch sizes of $10^3$, the results remain consistent. The average algorithmic speedup of $5.9\cdot 10^1$, again $\approx 10^3\times$ less than for single edge upates.
Compared to the results in Section~\ref{sub:exp_results_speedup}, we observe a larger variance in the algorithmic speedup when comparing to the parallel implementation of \staticb\ in ExaGraph.
This can be explained by the differences in the implementation of \staticb\ in NetworKit and ExaGraph.

\pagebreak

\subsection{Detailed Experimental Data}
\label{subsub:exp_results_additional}
\vspace{-0.8cm}
\begin{figure*}[h!]
    \centering
    \includegraphics[width=0.96\linewidth]{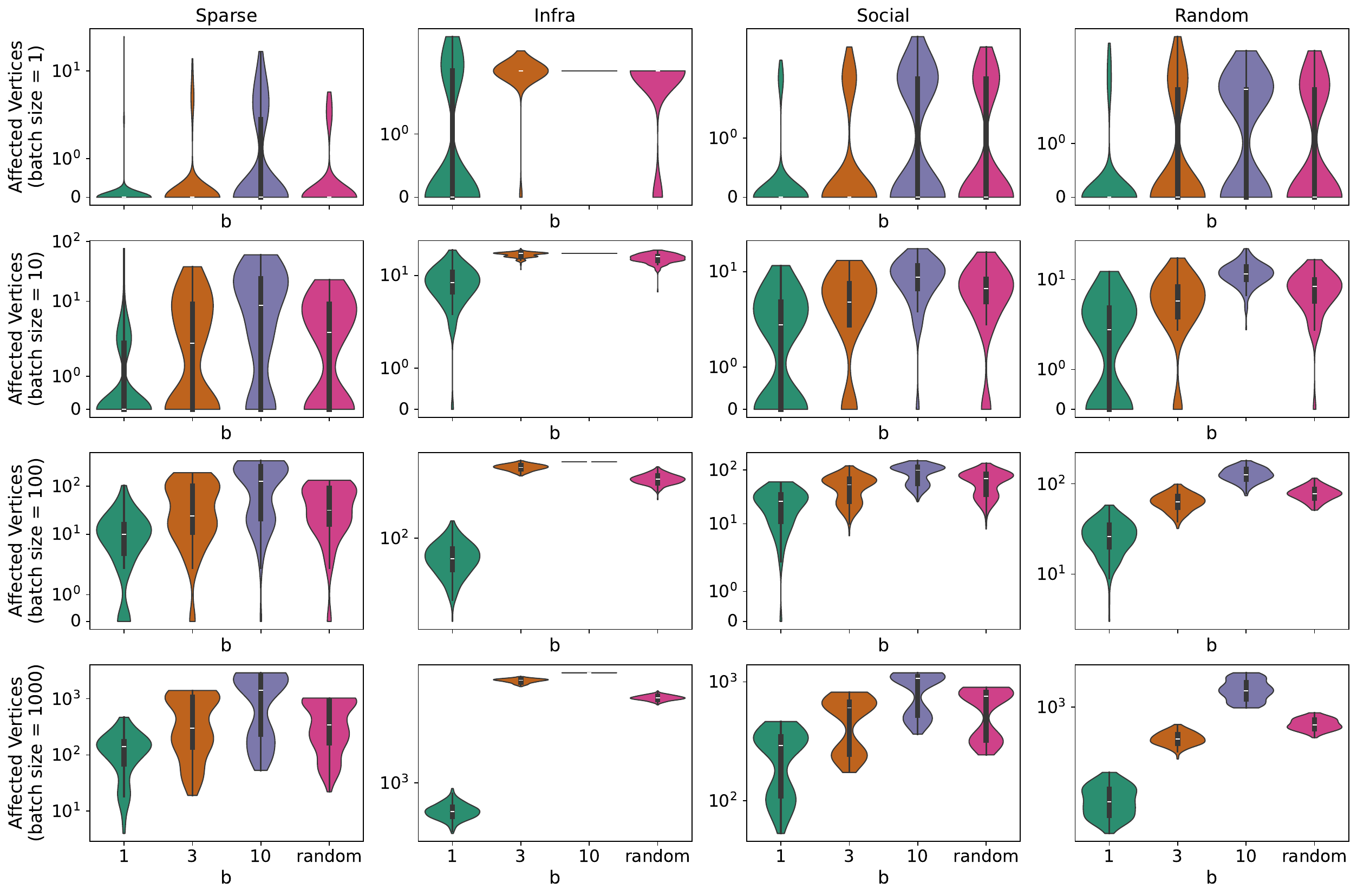}
    \vspace{-0.5cm}
    \caption{Number of \textbf{affected nodes} for \dynb\ for \textbf{edge insertion} for varying values of b ($1$, $2$, $3$, $10$, random) und batch sizes ($10^0$,$10^1$,$10^2$,$10^3$).}
    \label{fig:res_aff_insertion_misc}
\end{figure*}

\vspace{-0.5cm}

\begin{figure*}[h!]
    \centering
    \includegraphics[width=0.96\linewidth]{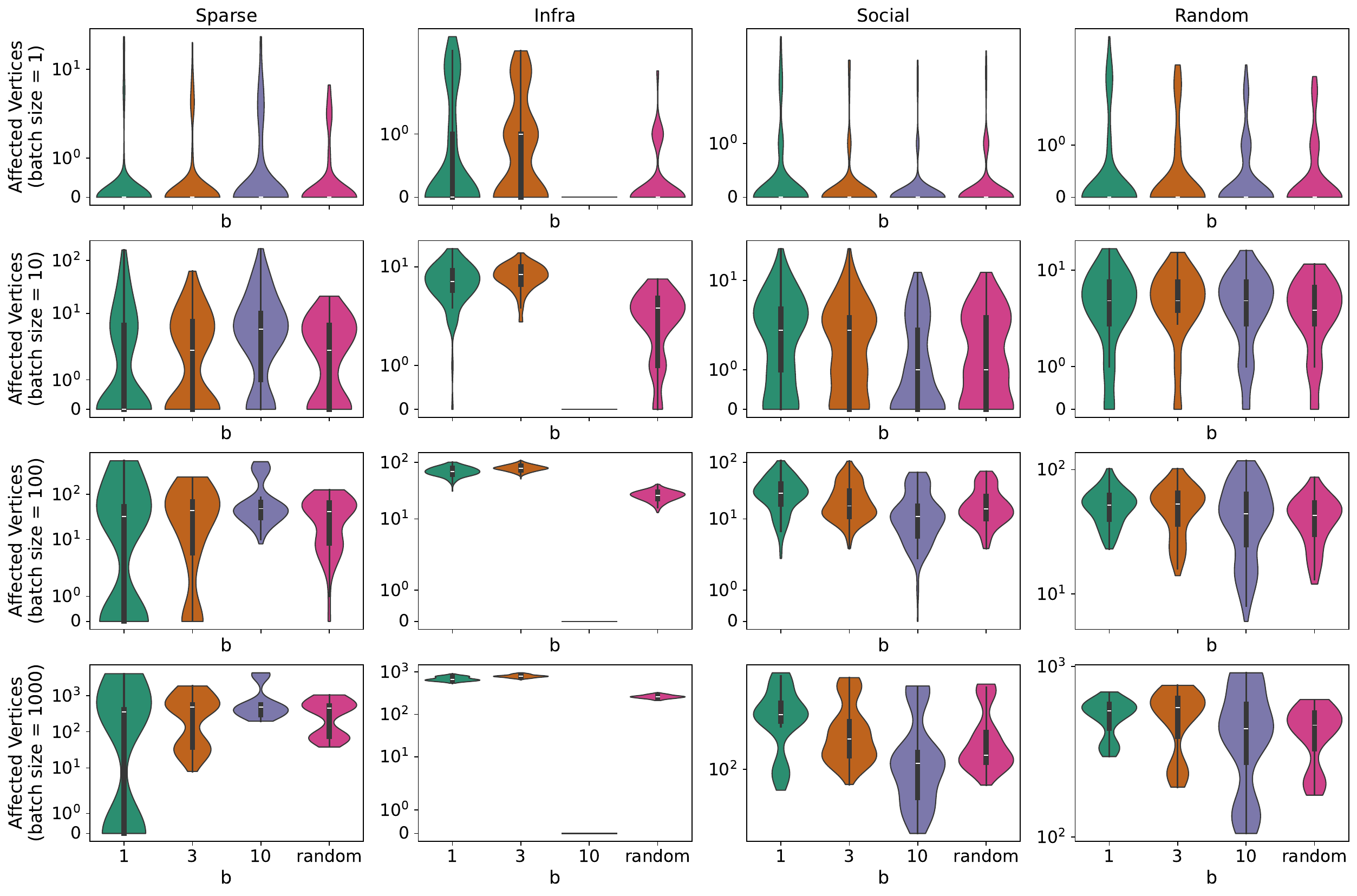}
    \vspace{-0.5cm}
    \caption{Number of \textbf{affected nodes} for \dynb\ for \textbf{edge removal} for varying values of b ($1$, $2$, $3$, $10$, random) und batch sizes ($10^0$,$10^1$,$10^2$,$10^3$).}
    \label{fig:res_aff_removal_misc}
\end{figure*}

\FloatBarrier

\begin{figure*}[t]
    \centering
    \includegraphics[width=\linewidth]{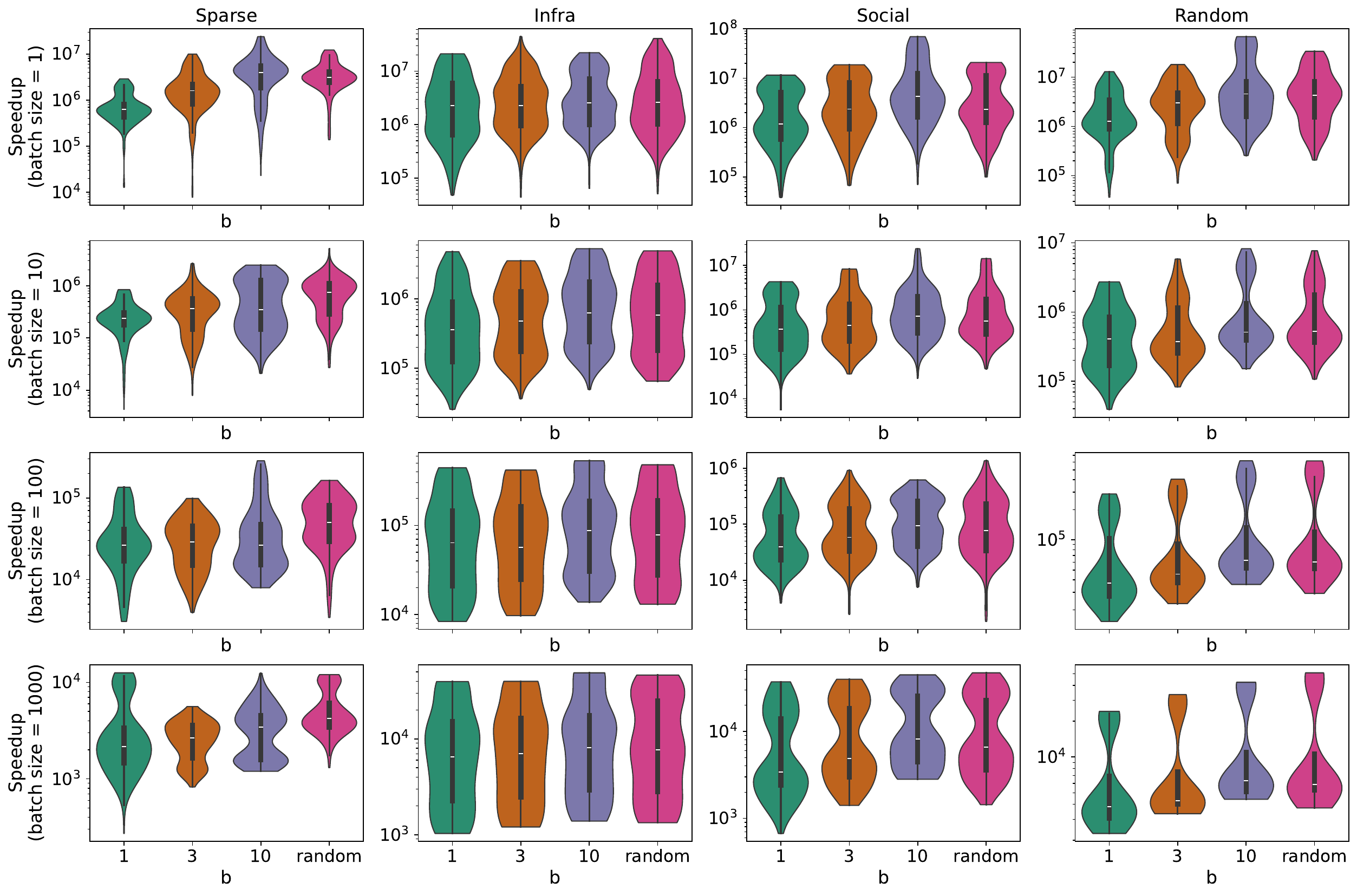}
    \caption{\textbf{Speedup} of \dynb\ vs. \staticb\ for \textbf{edge insertion} for varying values of b ($1$, $2$, $3$, $10$, random) und batch sizes ($10^0$,$10^1$,$10^2$,$10^3$).}
    \label{fig:res_speed_insertion_misc}
\end{figure*}

\begin{figure*}[b]
    \centering
    \includegraphics[width=\linewidth]{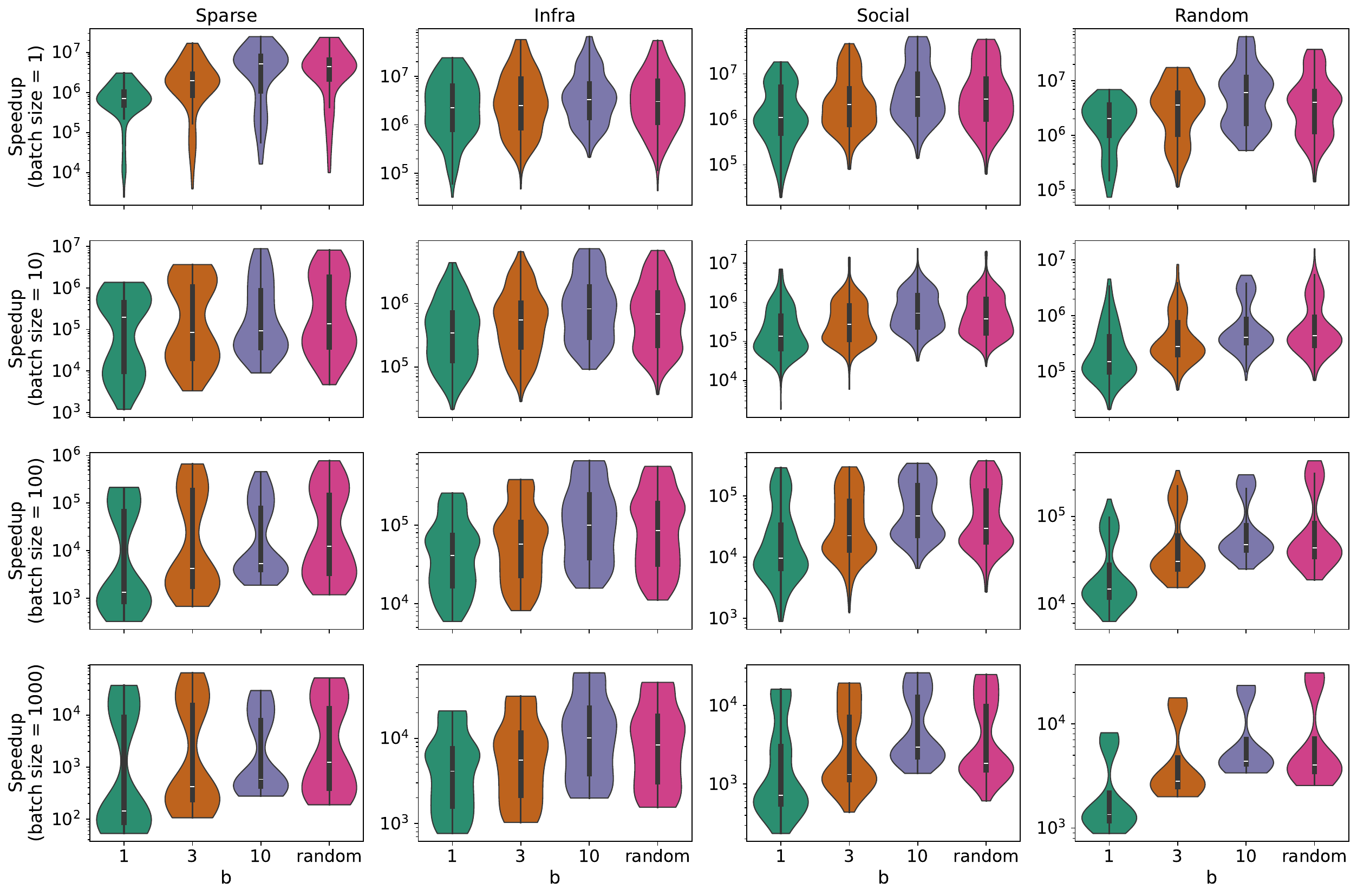}
    \caption{\textbf{Speedup} of \dynb\ vs. \staticb\ for \textbf{edge removal} for varying values of b ($1$, $2$, $3$, $10$, random) und batch sizes ($10^0$,$10^1$,$10^2$,$10^3$).}
    \label{fig:res_speed_removal_misc}
\end{figure*}

\pagebreak
\FloatBarrier
\pagebreak 

\begin{figure*}[h!]
    \centering
    \includegraphics[width=\linewidth]{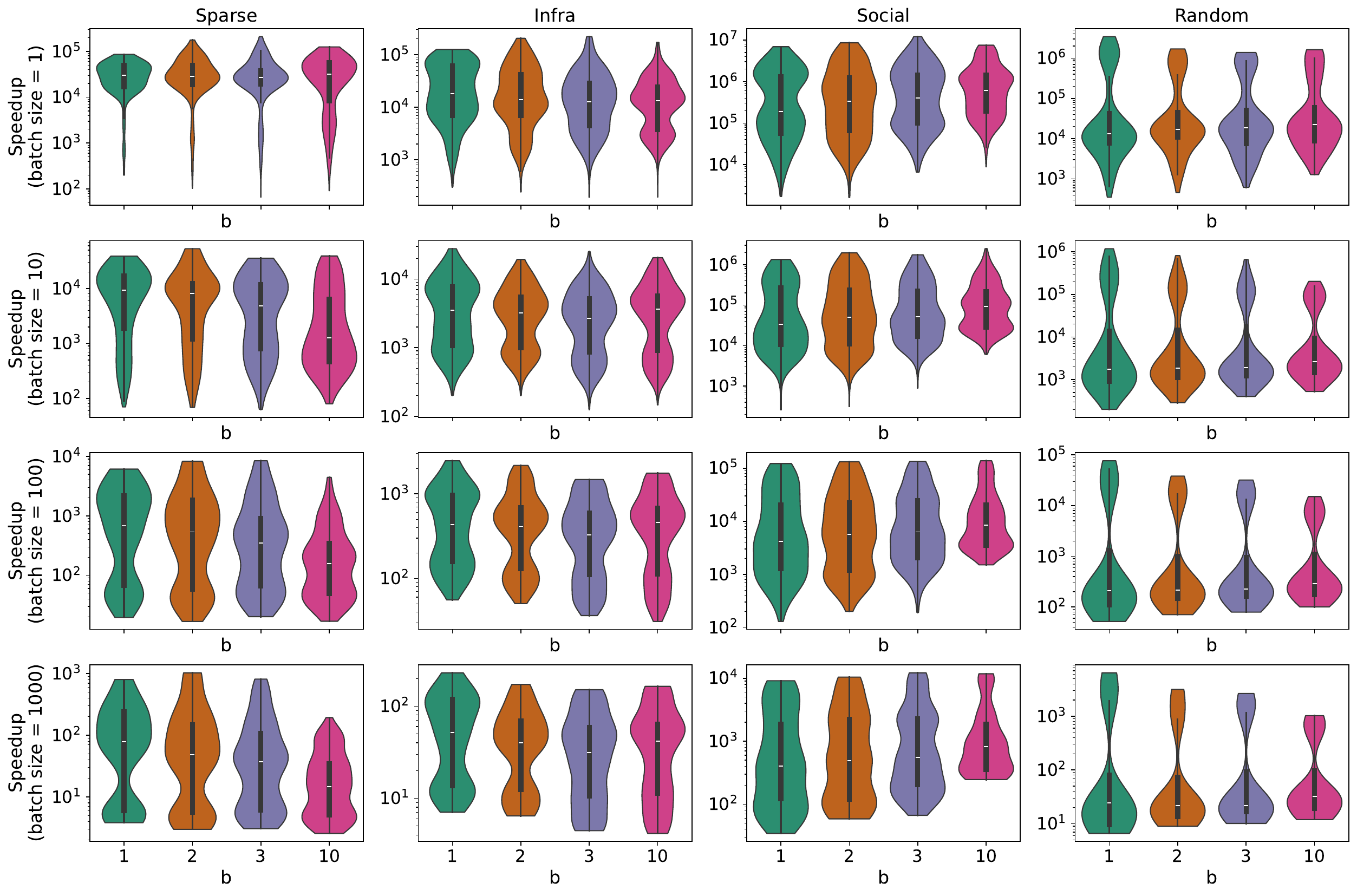}
    \caption{\textbf{Speedup} of \dynb\ against \textbf{parallel} \staticb\ for \textbf{edge insertion} for varying values of b ($1$, $2$, $3$, $10$) und batch sizes ($10^0$,$10^1$,$10^2$,$10^3$).}
    \label{fig:res_speed_par_misc}
\end{figure*}

\end{document}